\documentclass[12pt]{amsart}
\usepackage{amsfonts, amssymb}
\usepackage{amsmath}
\usepackage{amsthm}
\usepackage{hyperref}
\hypersetup{colorlinks=true
urlcolor=green
linkcolor=blue
citecolor=red}
\usepackage{xcolor}
\newtheorem{theorem}{Theorem}
\newtheorem{corollary}{Corollary}[theorem]
\newtheorem{lemma}{Lemma}
\newtheorem{remark}{Remark}

\begin{document}

\title[Integrable Differential Systems for Deformed Laguerre-Hahn O.P.]{Integrable Differential Systems for Deformed Laguerre-Hahn Orthogonal Polynomials}

\author{Maria das Neves Rebocho$^1$, Nicholas Witte$^{2,3}$}

\address{$^1$
Departamento de Matem\'atica and CMA-UBI, Universidade da Beira Interior, Rua Marqu\^{e}s d’\'Avila e Bolama, 6201-001 Covilh\~{a}, Portugal}
\address{$^2$ Department of Computer Science, Texas Tech University, Lubbock TX 79409-3104, USA}
\address{$^3$ School of Mathematics and Statistics, Victoria University of Wellington,Wellington 6012, New Zealand}
\email{\tt mneves@ubi.pt, n.s.witte@protonmail.com}

\date{\today}

\textcolor{red}{}
\begin{abstract}
Our work studies sequences of orthogonal polynomials $ \{P_{n}(x)\}_{n \geq 0}$ of the Laguerre-Hahn class,
whose Stieltjes functions satisfy a 
 Riccati type differential equation with polynomial coefficients,
{which} are subject to a deformation parameter $t$.
We derive systems of differential equations and give Lax pairs,
yielding non-linear differential equations in $t$ for the recurrence relation coefficients and Lax matrices of the orthogonal polynomials.
A specialisation to a non semi-classical case obtained via a M\"{o}bius transformation of a Stieltjes function related to a  deformed Jacobi weight is studied in detail,
showing this system is governed by a differential equation of the Painlev\'e type P$_\textrm{VI}$.
The particular case of P$_\textrm{VI}$ arising here has the same four parameters as the solution found by Magnus \cite{magnus-jcam} but differs in the boundary conditions.
\end{abstract}


\maketitle

\section{Introduction}
\subsection{Integrable Probability and Mathematical Physics}

Integrability, in one form or another, is central in soluble statistical mechanics models \cite{Baxter_1982},\cite{WMTB_1976},
the one-dimensional impenetrable Bose gas  \cite{FFGW_2003a}, \cite{JMMS_1980},
the enumerative combinatorics of the longest increasing subsequences of random permutations \cite{BDJ_1999},
the Gaussian, Laguerre and Jacobi unitary ensembles of random matrix theory \cite{TW_1994}, \cite{TW_1994b}, \cite{TW_1994a},
non-equilibrium particle hopping models such as the Asymmetric Simple Exclusion Process \cite{TW_2008}
and models in the Kardar-Parisi-Zhang universality class \cite{PS_2004}, \cite{SFP_2005},
and the modelling of L-functions in analytic number theory  \cite{CRS_2006}, \cite{KS_2000a}.
In particular, for the ensembles of random matrices mentioned above, which have joint eigenvalue probability distribution function given in terms of certain modifications of the classical Gaussian, Laguerre or Jacobi weights, it is well-known that certain statistical quantities related to the eigenspectrum can be expressed in terms of semi-classical polynomials or of the corresponding Hankel determinants, and are, in turn, characterised by  certain special  non-linear differential equations 
(see \cite[Chap. 8]{Forr_2010}, \cite{ble-its,F-W2,F-W3, F-W4} and their references).
Additional connections with integrable Hamiltonians, B{\"a}cklund transformations, Lax pairs, Toda lattices and Painlev\'e equations are well-known in the literature
and we refer the interested reader to the reference list of \cite[Sec. 1]{magnus-jcam}.

Integrability arises in these models and applications either explicitly as a semi-classical weight $ w $
(here a deformation of one of the classical weights in the Askey Tableaux \cite{Koekoek+Lesky+Swarttouw_2010} according to a precise definition given below),
or equivalently by $w$ being characterised as the solution of a Pearson equation - the first order homogeneous differential equation
\begin{equation}
	A(x)\partial_{x}w(x) = C(x)w(x)\,,
\label{PearsonEq}
\end{equation}
with $A, C$ being co-prime polynomials.
Alternatively, this characterisation may be implied by the fact that the Stieltjes transform of the weight $ w $ with support $I$
\begin{equation}
	f(z) := \int_{I}\frac{w(x)}{z-x}\,dx ,\qquad z \in \mathbb{C}\backslash I\,, \label{eq:Sf-w}
\end{equation}
satisfies the inhomogeneous form of the Pearson equation,
\begin{equation}
	A(x)\partial_{x} f(x) = C(x)f(x)+D(x) .
\label{StieltjesFn_SCclass}
\end{equation}
This latter form can also be taken as the defining characteristic \cite{shohat}.

However it is known within approximation theory that the semi-classical class is subsumed within a broader class - known as the Laguerre-Hahn class -
where the defining equation is the Riccati generalisation of \eqref{StieltjesFn_SCclass}, namely \cite{magnus-ric}
\begin{equation}
	A(x)\partial_{x}f(x) = B(x)f^2(x)+C(x)f(x)+D(x)\,.
\label{eq:ric-f}
\end{equation}
When the leading coefficient $ B$ vanishes we are back to the semi-classical class.
The class of Laguerre-Hahn orthogonal polynomials - the ones corresponding to Stieltjes functions satisfying (\ref{eq:ric-f}),
contains not only the semi-classical polynomials but also their standard modifications
\cite{ask,garcia-etal,gros1,gros2,peher-perturb,ronv-bel,ronv-marc,ronv-va,slim,zed}.
In the viewpoint of orthogonal polynomial theory, special emphasis has been given to the problem of classification, that is, describing the recurrence relation coefficients of the orthogonal polynomial system, given the polynomials $A, B, C, D$ in (\ref{eq:ric-f}) (see, amongst many others, \cite{filipuk-MNR-JMAA,filipuk-MNR-ITSF,magnus-snul}).

In the present paper we take Stieltjes functions satisfying (\ref{eq:ric-f}) subject to a dependence on an additional variable, say, a deformation parameter $t$.
In general terms, the main goal is to analyse the recurrence relation coefficients of the corresponding orthogonal polynomials under such a dependence.
These kind of problems have been analysed from different perspectives and we refer the interested reader to \cite{AB-MNR-deformed} and its references;
see also \cite{peher-toda}, where the equivalence between the equations of the infinite Toda (Volterra) lattice and the Riccati
differential equation for the Stieltjes function is established.

The framework of our study relies on the matrix differential systems satisfied by the orthogonal polynomials.
Here we adopt, to the Laguerre-Hahn setting,
a technique that can be traced back to the work of E. Laguerre in \cite{lag} - the so-called Laguerre method. Essentially, this method  gives a mechanical way to obtain differential equations satisfied by orthogonal
polynomials provided one knows the linear differential equation for the Stieltjes function. It has been revisited in recent literature on semi-classical polynomials, see, for instance,
 \cite{bonan} and \cite[Sec. 3]{magnus-jcam}, amongst many others.
 In the present paper, by using a similar algorithm to  \cite{lag}, we deduce that the  systems of orthogonal polynomials corresponding to (\ref{eq:ric-f}) satisfy difference-differential equations (see Theorem \ref{theo:magnus}),
which are then coupled into the form of matrix Sylvester equations for $ Y_{n}(x,t) $ defined in \eqref{eq:Yn-til} below
\begin{equation*}
	\partial_x Y_n=\mathcal{A}_nY_n-Y_n\mathcal{C}\,, \quad n \geq 0\,,
\end{equation*}
with $\mathcal{A}_n$ and $\mathcal{C}$  matrices with rational functions of $x$ as entries.
Similarly, one has a differential equation for $ f(x,t) $ with respect to the deformation variable $t$, given as
\begin{equation*}
	\widehat{A} \partial_t f=\widehat{B}f^2+\widehat{C}f+\widehat{D}\,, \label{eq:ric-dif-S-t}
\end{equation*}
and a system of   Sylvester equations
\begin{equation*}
	\partial_t Y_n=\mathcal{B}_nY_n-Y_n\widehat{\mathcal{C}}\,, \quad n \geq 0\,,
\label{eq:sylvester-Yn-t}
\end{equation*}
where $\mathcal{B}_n$ and $\widehat{\mathcal{C}}$ are  matrices with rational entries in the spectral variable $x$.
The consistency of the two differential equations for $Y_n$ combined with the consistency of the two differential equations for the Stieltjes function yields
(under the assumption that $f$ is not algebraic) the compatibility conditions deduced in Theorem \ref{teo:compat},
which are comparable to the Schlesinger equations in the semi-classical setting \cite{chud}.

In Section \ref{sec:3}, we specialize the results in Section \ref{sec:2} to a sequence of Laguerre-Hahn polynomials.
It is our aim to show by way of that
specific example that a natural deformation of an appropriate Laguerre-Hahn class system directly related to the Jacobi weight
does lead to the sixth Painlev{\' e} equation,
in fact to a solution of P$_{\textrm{VI}}$ with an identical parameter set to that of the deformed (semi-classical) Jacobi weight and only differing in the boundary conditions.
Concluding remarks are given in Section \ref{sec:4}.

\subsection{Main Results}

To present the main results of the paper some facts about  the sixth Painlev\'e equation and its Hamiltonian formulation are required.
Painlev\'e VI - usually denoted by P$_{\textrm{VI}}$ - is commonly presented as the  sixth  equation in the Painlev\'e list \cite{assche-book},
consisting of the six non-linear differential equations having the special property that all movable singularities are poles.

P$_{\textrm{VI}}$
conventionally refers to the four-parameter
second-order nonlinear differential equation
\begin{multline}
	y''=\frac{1}{2}\left(\frac{1}{y}+\frac{1}{y-1}+\frac{1}{y-t}\right)(y')^2-\left(\frac{1}{t}+\frac{1}{t-1}+\frac{1}{y-t}\right)y'\\
	+\frac{y(y-1)(y-t)}{t^2(t-1)^2}\left(\delta_1+\frac{\delta_2 t}{y^2}+\frac{\delta_3(t-1)}{(y-1)^2}+\frac{\delta_4t(t-1)}{(y-t)^2}\right)\,.
\label{eq:P6}
\end{multline}
Here, and in the sequel, the dash denotes derivative with respect to $t$.
In the Hamiltonian formulation of P$_{\textrm{VI}}$ the Hamiltonian $H(q,p, t; {v_k})$, where ${v_k}$ are parameters, is given by Okamoto \cite{oka} as
\begin{multline}
	t(t-1)H = q(q-1)(q-t)p^2-\left[(v_3+v_4)(q-1)(q-t)+(v_3-v_4)q(q-t)\right.\\
	\left.-(v_1+v_2)q(q-1)\right]p+(v_3-v_1)(v_3-v_2)(q-t)\,.
\label{eq:H6}
\end{multline}
The Okamoto theory in the Hamiltonian approach to P$_{\textrm{VI}}$ gives that, after eliminating $p$ in the Hamilton equations
\begin{equation}
	q'=\frac{\partial H}{\partial p}\,, \;\;p'=-\frac{\partial H}{\partial q}\,,
\label{eq:hamilton-eqs}
\end{equation}
then $q$ satisfies (\ref{eq:P6}) with parameters
\begin{equation*}
	\delta_1=\frac{1}{2}(v_1-v_2)^2\,,\; \delta_2=-\frac{1}{2}(v_3+v_4)^2\,,\; \delta_3=\frac{1}{2}(v_3-v_4)^2\,,\; \delta_4=\frac{1}{2}(1-(1-v_1-v_2)^2) \,.
\label{eq:par-v}
\end{equation*}

Let us now introduce some notations (the complete data is given in Sections \ref{sec:2} and \ref{sec:3}).
We start by considering the transformation giving a new Stieltjes function $\tilde{f}$ (see (\ref{eq:tilde-f}))
\begin{equation}
	\tilde{f}(x)=x-\beta_0-\frac{w_0}{f(x)}\,,
\label{eq:til-f}
\end{equation}
where $w_0$ is the zero moment of the {deformed} Jacobi weight
\begin{equation}
	w(x,t) = x^{\alpha}(1-x)^{\beta}(x- t)^{\mu}\,, \quad \alpha, \beta, \mu >-1\,,
\label{eq:jacobi-w}
\end{equation}
 the function $f$ is the Stieltjes transform of $w(x,t)$,
and $\beta_0$ is the recurrence coefficient of the sequence of monic orthogonal polynomials corresponding to $w(x,t)$.

The sequences of orthogonal polynomials corresponding to the weight (\ref{eq:jacobi-w}) were the subject of analysis in several works,
showing relations with Painlev\'e equations in the context of random matrix theory and linear isomonodromy problems.
We refer the interested reader to \cite{F-W1,magnus-jcam}.
In the present paper we are mainly concerned with the analysis of the orthogonal polynomials related to the Stieljtes function $\tilde{f}$ defined by (\ref{eq:til-f}).
Essentially, we will deduce differential systems for these polynomials, and analyse their dependence on the deformation parameter $t$ from the weight (\ref{eq:jacobi-w}).
At this stage it is important to emphasise that, to the best of our knowledge, a closed formula for the weight function $\tilde{w}$ corresponding to $\tilde{f}$ is not known.

We shall denote by $\{\tilde{P}_n\}_{n\geq 0}$  the sequence of monic orthogonal polynomials correspon\-ding to $\tilde{f}$,
and $\tilde{\beta}_n, \tilde{\gamma}_n$ its recurrence relation coefficients.
The matrix $\tilde{Y}_n$ given by
\begin{equation}\tilde{Y}_n=\left[
\begin{matrix}
	\tilde{P}_{n+1} & \tilde{P}_n^{(1)} \\
	\tilde{P}_n & \tilde{P}_{n-1}^{(1)} \\
\end{matrix}\right]\,,
\label{eq:Yn-til}
\end{equation}
where  $\{\tilde{P}_n^{(1)}\}$ is the sequence of monic associated polynomials (see {\eqref{eq:def-Pn1},} Section \ref{sec:2}), satisfies the matrix Sylvester equation
 \begin{equation*}
\partial_x {\tilde{Y}}_n={\tilde{\mathcal{A}}}_n\tilde{Y}_n-\tilde{Y}_n{\tilde{\mathcal{C}}}\,,
\end{equation*}
with
\begin{equation*}{\tilde{\mathcal{A}}}_n=\frac{1}{\tilde{A}}
	\begin{bmatrix}
	 \tilde{l}_{n} & \tilde{{\Theta}}_{n} \\
	-\tilde{\Theta}_{n-1}/{\tilde{\gamma}}_n &\;\;
	-\tilde{l}_{n}
	\end{bmatrix},
	\;\;{\tilde{\mathcal{C}}}=\frac{1}{\tilde{A}}
	\begin{bmatrix}
	\tilde{C}/2 & -\tilde{D}/\tilde{w}_0 \\
	\tilde{w}_0 \tilde{B} & -\tilde{C}/2
	\end{bmatrix}, \;\; 
\end{equation*}
with $\tilde{A}=x(x-1)(x-t)$ and $\tilde{B}, \tilde{C}, \tilde{D}$ the remaining  polynomial coefficients of the $x$-Riccati equation for $\tilde{f}$
(see (\ref{eq:eg-pol-B1})-(\ref{eq:eg-pol-D1})),
and where $\tilde{w}_0$ is the zero moment in the asymptotic expansion around $\infty $ of $\tilde{f}$.
The matrix $\tilde{\mathcal{A}}_n$ is parameterised in terms of the polynomials $\tilde{l}_n, \tilde{\Theta}_n$, which now have degrees two and one, respectively,
by virtue of the bounds $\deg(\tilde{B})=1, \deg(\tilde{C})=2, \deg(\tilde{D})=1$  (see Lemma \ref{lemma:data-systems}),
\begin{equation*}
	\tilde{l}_n(x)=\frac{(\nu_n-1)}{2}x^2+\tilde{l}_{n,1}x+\tilde{l}_{n,0}\,, \;\; \frac{\tilde{\Theta}_{n}(x)}{\tilde{\gamma}_{n+1}}=-\nu_n x- \vartheta_n\,,
\end{equation*}
with
\begin{equation*}
	\nu_n=2n+5+\alpha+\beta+\mu\,, \;\;\;
	\vartheta_n=\nu_n(\tilde{\beta}_{n+1}-1-t)+2\sum_{j=0}^{n}\tilde{\beta}_j+\tilde{\beta}_{n+1}+\mu t+\beta+2\beta_0\,.
\end{equation*}
As we shall see below in Section \ref{sec:3}, the zero of the function $\tilde{\Theta}_n$,
i.e., the unique zero of the  $ (1,2)$ entry of the matrix $\tilde{\mathcal{A}}_n$ from the corresponding Sylvester equations,
that is, $q:=-\vartheta_n/\nu_n$, will be the Painlev\'e transcendent.

Let us also define  $\varphi$ - a polynomial in the variable $x$ - in terms of the polynomials $\tilde{A}, \tilde{B}, \tilde{C}, \tilde{D}$ and $\tilde{w}_0$, as
\begin{equation*}
	\varphi=\frac{\tilde{D}}{\tilde{w}_0}\left(\tilde{A}+\tilde{w}_0\tilde{B}\right)-\frac{\tilde{C}^2}{4}\,.
\end{equation*}
The evaluation of $\varphi$ at the zeroes of $\tilde{A}$ will have a fundamental role in the derivation of P$_{\textrm{VI}}$.
Our main result can now be stated as follows.
\begin{theorem}\label{teo:P6}
Under the above notations, let $v_1, v_2, v_3 , v_4$ be given in terms of the function $\varphi $ as 
\begin{equation}\label{eq:parameter-varphi}
	(v_1-v_2)^2 = 1\,, \;
	(1-v_1-v_2)^2 = -4\frac{\varphi(t)}{t^2(t-1)^2}\,,\;
	(v_3+v_4)^2 = -4\frac{\varphi(0)}{t^2}\,,\;
	(v_3-v_4)^2 = -4\frac{\varphi(1)}{(t-1)^2}\,. \;
\end{equation}
In terms of the coefficients $\tilde{\beta}_n, \tilde{\gamma}_n$ in the three-term recurrence relation satisfied by the sequence of monic orthogonal polynomials $\{\tilde{P}_n\}_{n\geq 0}$, let $\nu_n$ and  $\vartheta_n$ be defined by the above equations,
\begin{align*}
	\nu_n&=2n+5+\alpha+\beta+\mu\,,
\\
	\vartheta_n&=\nu_n(\tilde{\beta}_{n+1}-1-t)+2\sum_{j=0}^{n}\tilde{\beta}_j+\tilde{\beta}_{n+1}+\mu t+\beta+2\beta_0\,.
\end{align*}
The Hamilton equations are satisfied by
\begin{equation}
	q = -\vartheta_n/\nu_n\,,
\label{eq:hamilt-q}
\end{equation}
and $p$, which is given by
\begin{equation}
	p = \frac{\nu_nq^2+2q\tilde{l}_{n,1}-q+2\tilde{l}_{n,0}}{2q(q-1)(q-t)}+\frac{v_3+v_4}{2q}+\frac{v_3-v_4}{2(q-1)}-\frac{v_1+v_2}{2(q-t)}\,.
\label{eq:2nd-hamilt-p}
\end{equation}
\end{theorem}

\begin{remark}
It turns out that the quantities $\varphi(0), \varphi(1), \varphi(t)$ are such that
\begin{equation*}
	\;\;\frac{\varphi(0)}{t^2} = -\tfrac{1}{4}\alpha^2\,,
	\;\;\frac{\varphi(1)}{(t-1)^2} = -\tfrac{1}{4}\beta^2\,,
	\;\;\frac{\varphi(t)}{t^2(t-1)^2} = -\tfrac{1}{4}\mu^2\,,
\end{equation*}
and hence, $q$ satisfies (\ref{eq:P6}) with parameters
$(\delta_1, \delta_2, \delta_3, \delta_4)= (\tfrac{1}{2}, -\tfrac{1}{2}\alpha^2,  \tfrac{1}{2}\beta^2, \tfrac{1}{2}(1-\mu^2))$.
\end{remark}

\section{Lax Pairs for Laguerre-Hahn Orthogonal Polynomials}\label{sec:2}

\subsection{Orthogonal polynomials, Stieltjes functions, and the Laguerre-Hahn class}
Aiming at giving a more self-contained overview from the viewpoint of orthogonal polynomials, we now  introduce further notations and preliminary results to be used in the sequel.

We shall consider weight functions (possibly deformed and complex - however this is not relevant for our current argument and is notationally suppressed)
supported on an interval $ I\subset \mathbb{R} $, for which all the moments
$$w_n=\int_{I}x^nw(x)\, dx\,, \quad n\geq 0\,,$$ exist and satisfy the regularity condition that the Hankel determinants
   $$\Delta_n:=\det[w_{j+k-2}]_{j, k=1, \dots, n}\,, \quad \Delta_0:=1\,, $$ are such that \begin{equation*}\Delta_n \neq 0, \;\;\; n \geq 0\,.\label{eq:hankel}\end{equation*}

The corresponding sequences of orthogonal polynomials, which we will take monic,
$\{P_n(x)=x^n+\textrm{lower degree terms}\}_{n \geq 0}$, defined via the orthogonality relation
\begin{equation}
	\int_I P_n(x)P_m(x)w(x)\,dx=h_n\delta_{m,n}\,,\quad n,m \geq 0\,, \;\; h_n \neq 0\,, \quad h_n =\displaystyle\frac{\Delta_{n+1}}{\Delta_{n}}\,, \quad
\label{eq:ort}
\end{equation}
are fully characterised by the three-term recurrence relation \cite{szego}
\begin{equation}
	xP_n(x)=P_{n+1}(x)+\beta_nP_n(x)+\gamma_nP_{n-1}(x)\,,\quad n = 1, 2, \dots\,,
\label{eq:ttrr-Pn}
\end{equation}
with the initial conditions $P_0(x)=1, \;P_1(x)=x-\beta_0$.
The numbers $\beta_n, \gamma_n,$ commonly called the recurrence relation coefficients, can be expressed as
\begin{equation}
	\beta_n=\frac{1}{h_n}\int_I xP_n^2(x) w(x)\,dx\,, \;  n\geq 0\,, \;\;\;
	\gamma_{n}=\frac{1}{h_{n-1}}\int_I P^2_{n}(x)w(x)\,dx\,, \;\; n\geq 1\,.
\label{eq:3trr}
\end{equation}
Thus, the norms $h_n$ in  (\ref{eq:ort}) satisfy
\begin{equation}
	h_0=w_0\,, \;\; \gamma_n=h_n/h_{n-1}\,, \; n \geq 1\,,
\label{eq:gamma-hn}
\end{equation}
hence, the recurrence coefficients $\gamma_n$ are given in terms of the Hankel determinants as
\begin{equation*}
	\gamma_n=\frac{\Delta_{n-1}\Delta_{n+1}}{\Delta^2_{n}}\,, \quad n \geq 1\,.
\label{eq:gamma-via-det}
\end{equation*}

Given the sequence of monic polynomial $\{P_n\}_{n\geq 0}$ corresponding to a weight $w$,  one defines the sequence of monic associated polynomials,
denoted by $\{P_n^{(1)}\}_{n\geq 0}$, as \cite{chihara,szego}
\begin{equation}
	P_n^{(1)}(x) := \frac{1}{w_0}\int_{0}^{1}\frac{P_{n+1}(x)-P_{n+1}(y)}{x-y}w(y)\,dy\,, \quad n=0,1, \dots \,.
\label{eq:def-Pn1}
\end{equation}
As an immediate consequence of (\ref{eq:ttrr-Pn}), it follows that $\{P_n^{(1)}\}_{n\geq 0}$ satisfies the three-term recurrence relation
\begin{equation}
	xP^{(1)}_{n-1}(x) = P^{(1)}_{n}(x)+\beta_nP^{(1)}_{n-1}(x)+\gamma_nP^{(1)}_{n-2}(x)\,, \quad n=1, 2, \dots \,,
\label{eq:ttrr-Pn1}
\end{equation}
with  initial conditions $P^{(1)}_{-1}(x)=0, \;P^{(1)}_0(x)=1$.
The following relation is well-known \cite{szego}:
\begin{equation}
	P_n^{(1)}(x)P_n(x)-P_{n+1}(x)P_{n-1}^{(1)}(x)=\prod_{k=1}^{n}\gamma_k\,, \quad n\geq 1\,.
\label{eq:liou}
\end{equation}

The  Stieltjes transform of the weight is a  central object in the theory. In the most general situation, it is taken as the formal series defined in terms of the moments as
\begin{equation}
	f(x)=f(w,x):=\displaystyle\sum_{n=0}^{+\infty}\frac{w_n}{x^{n+1}}\,.
\label{eq:expan-S}
\end{equation}

The polynomials $P_n, P_n^{(1)}$ and the Stieltjes function $f$ are related via \cite{szego}
\begin{equation*}
	P_{n+1}(x)f(x)-w_0P_n^{(1)}(x) = \mathcal{O}(x^{-n-1})\,, \quad x\to \infty,\;\;\; n \geq 0\,.
\label{eq:Herm-Pade}
\end{equation*}
This formula is commonly referred to as the Pad\'e formula, 
and will be written as
\begin{equation}
	P_{n+1}(x)f(x)-w_0P_n^{(1)}(x) =: \varepsilon_{n+1}(x)\,, \;\;\; n \geq 0\,.
\label{eq:H-P}
\end{equation}
Note that the functions $\varepsilon_n$ 
satisfy the same recurrence relation (\ref{eq:ttrr-Pn1}) with $\varepsilon_0=f$.
The interplay between orthogonal polynomials, the associated polynomials, and the corresponding Stieltjes function has been analysed in \cite{assche-survey}.

The function $f$ has a Jacobi continued fraction expansion given by (see \cite[pp. 85, eq. (4.1)]{chihara} combined with the notation $w_0=\gamma_0$ from \cite[Th. 4.4., Ch. I]{chihara})
\begin{equation*}
	f(x)=\frac{w_0}{x-\beta_0-\displaystyle\frac{\gamma_1}{x-\beta_1-\frac{\gamma_2}{\ddots}}}\,,
\label{eq:fcont-S}
\end{equation*}
where the $\gamma_{n}$'s and the $\beta_{n}$'s are the recurrence relation coefficients of the corresponding sequence of monic orthogonal polynomials.
Recalling that the continued fraction expansion for the Stieltjes function corresponding to the sequence $\{P_n^{(1)}\}_{n\geq 0}$ is
\begin{equation*}
	\tilde{f}(x)=\frac{\gamma_1}{x-\beta_1-\displaystyle\frac{\gamma_2}{x-\beta_2-\frac{\gamma_3}{\ddots}}}\,,
\end{equation*}
an immediate relation between $f$ and $\tilde{f}$ is
\begin{equation}
f(x)=\frac{w_0}{x-\beta_0-\tilde{f}(x)}\,.
\label{eq:tilde-f}
\end{equation}
Also, we note the asymptotic expansion of $\tilde{f}$ in terms of the corresponding moments,
\begin{equation}
	\tilde{f}(x)=\displaystyle\sum_{n=0}^{+\infty}\frac{\tilde{w}_n}{x^{n+1}}\,.
\label{eq:expan-S1}
\end{equation}
A consequence of (\ref{eq:tilde-f}), taking into account the asymptotic expansions of $f$ and $\tilde{f}$ given by (\ref{eq:expan-S}) and (\ref{eq:expan-S1}), is the relation for the moments
\begin{equation}
	\tilde{w}_n=\frac{w_{n+2}-\beta_0 w_{n+1}-\sum_{k=1}^{n}w_k \tilde{w}_{n-k}}{w_0}\,, \quad n \geq 0\,,
\label{eq:moments}
\end{equation}
with the convention $\sum_{i}^{j}\cdot=0$ whenever $i >j$.

In the present work we deal with Stieltjes functions  belonging to the so-called Laguerre-Hahn class that, under the conditions and definitions above,
 satisfy  a Riccati type differential equation (\ref{eq:ric-f}) with
 polynomial  coefficients. The transformation $\tilde{f}$ studied in Section~3 is a member of such a class, its Laguerre-Hahn character is displayed from equations
(\ref{eq:ric-f1-x})-(\ref{eq:eg-pol-D1}).
  In the framework of orthogonal polynomials, the Laguerre-Hahn class has been analysed from different perspectives.
Early refe\-rences on Laguerre-Hahn Stieltjes functions put emphasis on the connection between the Riccati equation
and the distributional equation for the corresponding linear functional - the so-called Laguerre-Hahn linear functionals or simply Laguerre-Hahn forms  (see \cite[Chap. 1]{MNR-survey}). Another direction of study proceeds with the analysis of differential equations of the Riccati type equation
(\ref{eq:ric-f})
focusing on transformations of orthogonal polynomials corresponding to rational Stieltjes spectral transformations  \cite{dehesa-etal,zed} (see also \cite{magnus-ric,peher-perturb}).
Further details on these topics may be found in the survey paper \cite{MNR-survey}.

\subsection{The Laguerre Method}\label{sec:2.1}

To pursue our goals we start by
deducing difference-differential relations involving the sequences of orthogonal polynomials and associated polynomials corresponding to Laguerre-Hahn Stieltjes functions.
Here, we  adopt the so-called Laguerre method \cite{lag} to obtain the Equations (\ref{eq:eqdifPn}) and (\ref{eq:eqdifPn1}) below.

\begin{theorem}
\label{theo:magnus}
Let $f$ be a Stieltjes function given by (\ref{eq:expan-S}),
let $\{P_n\}_{n \geq 0}$ and $\{P_n^{(1)}\}_{n\geq 0}$ be the corresponding sequences of monic orthogonal polynomials and associated polynomials, respectively.
If $f$ satisfies
\begin{equation}
	A\partial_{x}f=Bf^2+Cf+D\,,
\label{eq:ric-dif-f-x}
\end{equation}
with $A, B, C ,D $ co-prime polynomials, then the difference-differential equations hold, for all $n \geq 0$,
\begin{gather}
	A\partial_{x}P_{n+1}=(l_n-C/2)P_{n+1}-w_0BP_n^{(1)}+\Theta_nP_{n}\,,
\label{eq:eqdifPn}\\
	A\partial_{x}P^{(1)}_{n}=\frac{D}{w_0}P_{n+1}+(l_n+C/2)P^{(1)}_{n}+ \Theta_{n}P^{(1)}_{n-1}\,,
\label{eq:eqdifPn1}
\end{gather}
with $l_n, \Theta_n$ polynomials of $x$ with degrees given by
\begin{equation*}
	\deg(l_n)=\max\{\max\{\deg(A), \deg(B)\}-1, \deg(C)\}\,,\quad \deg(\Theta_n)=\deg(l_n)-1,
\end{equation*}
and satisfying the initial conditions
\begin{equation}
	\frac{\Theta_{-1}}{w_0}=D\,,\; \Theta_0=A+(x-\beta_0)(\frac{C}{2}-l_0)+w_0B\,,\;  l_{-1}=\frac{C}{2}\,,\; l_0=-\frac{C}{2}-(x-\beta_0)\frac{D}{w_0}\,.
\label{eq:initialcond}
\end{equation}
Here, $\beta_0$ is the recurrence relation coefficient of $\{P_n\}_{n \geq 0}.$
\end{theorem}
\begin{proof}
Using (\ref{eq:H-P}) in  (\ref{eq:ric-dif-f-x}) we
obtain, after standard computations,
\begin{multline}
	 A\left(\partial_{x}P_{n}^{(1)} P_{n+1}-\partial_{x}P_{n+1} P^{(1)}_{n}\right)-Bw_0(P^{(1)}_{n})^2-CP^{(1)}_{n}P_{n+1}-\frac{D}{w_0}P_{n+1}^2\\
	 =-\frac{A}{w_0}\left(\partial_{x}\varepsilon_{n+1} P_{n+1}-\partial_{x}P_{n+1} \varepsilon_{n+1}\right)+2BP_n^{(1)}\varepsilon_{n+1}+\frac{C}{w_0}\varepsilon_{n+1}P_{n+1}\,.
\label{eq:aux1-theo-magnus}
\end{multline}
Thus, the left hand side of the previous equation is a polynomial, say $\check{\Theta}_n$,
\begin{equation}
	A\left(\partial_{x}P_{n}^{(1)} P_{n+1}-\partial_{x}P_{n+1} P^{(1)}_{n}\right)-Bw_0(P^{(1)}_{n})^2-CP^{(1)}_{n}P_{n+1}-\frac{D}{w_0}P_{n+1}^2=\check{\Theta}_n\,,
\label{eq:aux2-theo-magnus}
\end{equation}
and its degree is bounded by a constant, as
  $$\check{\Theta}_n=-\frac{A}{w_0}\left(\partial_{x}\varepsilon_{n+1} P_{n+1}-\partial_{x}P_{n+1} \varepsilon_{n+1}\right)+2BP_n^{(1)}\varepsilon_{n+1}+\frac{C}{w_0}\varepsilon_{n+1}P_{n+1}$$
is bounded by a power less or equal than $\max\{\deg(A)-2, \deg(B)-2,\deg(C)-1\}$.

From (\ref{eq:liou}), replace $\check{\Theta}_n$ in (\ref{eq:aux2-theo-magnus}) by
$\displaystyle \frac{P_n^{(1)}P_n-P_{n+1}P_{n-1}^{(1)}}{\prod_{k=1}^{n}\gamma_k}\check{\Theta}_n$, and after rearranging and using the notation $\Theta_n=-\check{\Theta}_n/(\prod_{k=1}^{n}\gamma_k)$, we get
\begin{multline}
	P_n^{(1)}\left[A\partial_{x}P_{n+1}+\frac{C}{2}P_{n+1}-\Theta_nP_n+w_0BP_n^{(1)}\right]\\=
	P_{n+1}\left[A\partial_{x}P_{n}^{(1)}-\frac{C}{2}P_{n}^{(1)}-\Theta_nP_{n-1}^{(1)}-\frac{D}{w_0}BP_{n+1}\right]\,.
\label{eq:aux3-theo-magnus}
\end{multline}
As $P_n^{(1)}$ and $P_{n+1}$ do not have common zeroes (this property follows, for instance, from (\ref{eq:liou})),
then both hand sides of (\ref{eq:aux3-theo-magnus}) must therefore have the form $l_nP_n^{(1)}P_{n+1}$, where $l_n$ is a new auxiliary polynomial of bounded degree, thus we have
\begin{gather}
	A\partial_{x}P_{n+1}+\tfrac{1}{2}CP_{n+1}-\Theta_nP_n+w_0BP_n^{(1)}=l_nP_{n+1}\,,
\label{eq:aux5-theo-magnus}
\\
	A\partial_{x}P_{n}^{(1)}-\tfrac{1}{2}CP_{n}^{(1)}-\Theta_nP_{n-1}^{(1)}-\frac{D}{w_0}BP_{n+1}=l_nP_n^{(1)}\,,
\label{eq:aux6-theo-magnus}
\end{gather}
hence we get (\ref{eq:eqdifPn})-(\ref{eq:eqdifPn1}).

Note that, using again (\ref{eq:liou}) in (\ref{eq:aux5-theo-magnus})-(\ref{eq:aux6-theo-magnus}), we deduce the polynomial $l_n$ is given by
\begin{multline}
	\left(\prod_{k=1}^{n}\gamma_k\right)l_n = A\left(\partial_{x}P_n^{(1)}P_n-\partial_{x}P_{n+1}P_{n-1}^{(1)}\right)
	-\tfrac{1}{2}C\left(P_n^{(1)}P_n+P_{n+1}P_{n-1}^{(1)}\right)\\
	-\frac{D}{w_0}P_{n+1}P_n-w_0BP_n^{(1)}P_{n-1}^{(1)}\,.
\label{eq:ln-def}
\end{multline}
This concludes the proof.
\end{proof}

Using the three-term recurrence relations of $\{P_n\}_{n \geq 0}$ and  $\{P^{(1)}_n\}_{n \geq 0}$,
(\ref{eq:ttrr-Pn}) and (\ref{eq:ttrr-Pn1}), in (\ref{eq:eqdifPn})-(\ref{eq:eqdifPn1}), two further equations that
give $A\partial_{x}P_{n+1}$ and $A\partial_{x}P^{(1)}_{n}$ in terms of $P_n, P_{n-1}, P^{(1)}_{n-1} $ and $P^{(1)}_{n-2} $ are obtained,
and such four equations yield a system that can be coupled into the form of matrix Sylvester equations, as given in the next theorem. 

\begin{theorem}\label{teo:sylv-x}
(see \cite{AB-sylv}) Let the notations of the previous theorem hold.
Let the Stieltjes function $f$  given by (\ref{eq:expan-S}) satisfy the Riccati type differential equation (\ref{eq:ric-dif-f-x}), $A\partial_{x}f=Bf^2+Cf+D$.
Then, the following matrix Sylvester equations hold:
\begin{equation}
	\partial_{x}Y_n=\mathcal{A}_nY_n-Y_n\mathcal{C}\,, \quad n \geq 0\,, \quad
	 Y_n=\left[
	\begin{matrix}
		  P_{n+1} & P_n^{(1)} \\
		  P_n & P_{n-1}^{(1)} \\
	\end{matrix}\right]\,,
\label{eq:sylvester-Yn-x}
\end{equation}
where
\begin{equation}
	\mathcal{A}_n=\frac{1}{A}
	\begin{bmatrix}
		 l_{n} & \Theta_{n} \\
		-\Theta_{n-1}/\gamma_n &\;\;
		l_{n-1}+(x-\beta_n)\Theta_{n-1}/\gamma_n
	\end{bmatrix}\,,
	\;\;\mathcal{C}=\frac{1}{A}
	\begin{bmatrix}
		 C/2 & -D/w_0 \\
		w_0B & -C/2
	\end{bmatrix}\,.
\end{equation}
\end{theorem}

The differential systems of type (\ref{eq:sylvester-Yn-x}) will play a relevant role in the sequel.
Before continuing, let us remark the differences regarding the semi-classical case. Let us use the notation
\begin{equation}\breve{\mathcal{A}}_n=\begin{bmatrix}
 l_{n} & \Theta_{n} \\
-\Theta_{n-1}/\gamma_n &\;\;
l_{n-1}+(x-\beta_n)\Theta_{n-1}/\gamma_n\end{bmatrix}\,.
\label{eq:breve-An}
\end{equation}
Firstly, we stress that the sequences
$ \mathcal{Q}_n:=\left[
\begin{array}{cc}
  \varepsilon_{n+1} \\
  \varepsilon_{n}
\end{array}
\right], \; n \geq 0$,
with the $\varepsilon_n$ defined in (\ref{eq:H-P}), corresponding to Stieltjes functions such that  (\ref{eq:ric-dif-f-x}) holds,
satisfy the differential system \cite[Th. 1]{AB-sylv}
\begin{equation}
	A\partial_{x}\mathcal{Q}_n=\left(\breve{\mathcal{A}}_n+(Bf+C/2)I\right)\mathcal{Q}_n\,, \quad n \geq 0\,,\label{eq:Qn-LH}
\end{equation}
where $I$ denotes the $2\times 2$ identity matrix.
In the semi-classical case, taking $B \equiv 0$ in  (\ref{eq:Qn-LH}) and combining with the spectral differential equation for the weight, $ \partial_{x}\ln w=C/A,$
then a second solution is obtained,
\begin{equation*}
	A\partial_{x}\left[
	\begin{array}{c}
		\varepsilon_{n+1}/w \\
 		\varepsilon_n/w
	\end{array}\right]
	= \left(\breve{\mathcal{A}}_n-\frac{C}{2}I\right)
	\left[
	\begin{array}{c}
			\varepsilon_{n+1}/w \\
			\varepsilon_n/w
	\end{array}\right]\,,
\end{equation*}
which gives
\cite[Sec. 3]{magnus-jcam}
\begin{equation}
	\partial_{x}\left[
	\begin{array}{cc}
		P_{n+1} & \varepsilon_{n+1}/w \\
  		P_n & \varepsilon_n/w
	\end{array}\right]
	= \frac{1}{A}\left(\breve{\mathcal{A}}_n-\frac{C}{2}I\right)
	\left[\begin{array}{cc}
		P_{n+1} & \varepsilon_{n+1}/w \\
		P_n & \varepsilon_n/w
	\end{array}\right]\,.
\label{eq:sist-magnus}
\end{equation} 
Systems of the type (\ref{eq:sist-magnus}) are a standard tool when studying semi-classical orthogonal polynomials (see \cite{ble-its,F-W1} amongst many others).
Let us emphasize that, in the non semi-classical Laguerre-Hahn case, in general, explicit formulae for the weights are not available.
Indeed, it is an open problem to get the general form of the weight given a Riccati differential equation for the Stieltjes function (see \cite{zed}).

\subsection{Bilinear Recurrence Relations}
It is useful to write the recurrence relations (\ref{eq:ttrr-Pn}) and (\ref{eq:ttrr-Pn1}) in the matrix form
\begin{equation}
	Y_n=\mathcal{M}_nY_{n-1}\,, \quad \mathcal{M}_n=
	\begin{bmatrix}
		x-\beta_n & -\gamma_n \\
		1 & 0 \\
	\end{bmatrix}\,, \quad n \geq 1\,,
\label{eq:3trr-mat}
\end{equation}
with initial conditions
$Y_0=\begin{bmatrix}
	x-\beta_0 & -\gamma_0 \\
	1 & 0 \\
\end{bmatrix}$.
The matrix $\mathcal{M}_n$ is known as the transfer matrix. As described in \cite{AB-sylv},
the consistency of (\ref{eq:sylvester-Yn-x}) and (\ref{eq:3trr-mat}) yields the equation
\begin{equation}\label{eq:Ansylv}
	\partial_{x}\mathcal{M}_n={\mathcal{A}}_n\mathcal{M}_n-\mathcal{M}_n{\mathcal{A}}_{n-1}\,,\quad n \geq 1\,.
\end{equation}
Equation (\ref{eq:Ansylv}) implies two non-trivial equations from entries $(1,1)$ and $(1,2)$, respectively:
\begin{align}
	A+\gamma_n\frac{\Theta_{n-2}}{\gamma_{n-1}}-\Theta_n&=(x-\beta_n)(l_n-l_{n-1})\,,
\label{eq:pos-11}\\
	l_n+(x-\beta_n)\frac{\Theta_{n-1}}{\gamma_n}&=l_{n-2}+(x-\beta_{n-1})\frac{\Theta_{n-2}}{\gamma_{n-1}}\,.
\label{eq:pos-12}
\end{align}
After basic computations (\ref{eq:pos-12}) implies
\begin{equation}
	l_n+l_{n-1}=-(x-\beta_n)\frac{\Theta_{n-1}}{\gamma_n}\,.
\label{eq:aux-tr-2}
\end{equation}
The multiplication of (\ref{eq:pos-11}) and (\ref{eq:aux-tr-2}) gives us
\begin{equation*}
	-l_n^2+l_{n-1}^2=A\frac{\Theta_{n-1}}{\gamma_{n}}+\Theta_{n-1}\frac{\Theta_{n-2}}{\gamma_{n-1}}-\Theta_{n}\frac{\Theta_{n-1}}{\gamma_{n}}\,.
\end{equation*}
Summing on $n$ and using the initial conditions (\ref{eq:initialcond}), we get
\begin{equation}
	-l_n^2+\Theta_{n}\frac{\Theta_{n-1}}{\gamma_{n}}=\frac{D}{w_0}\left(A+w_0B\right)-\frac{C^2}{4}+A\sum_{k=1}^{n}\frac{\Theta_{k-1}}{\gamma_{k}}\,.
\label{eq:aux-det}
\end{equation}
Note that equation (\ref{eq:aux-tr-2}) reads as
\begin{equation*}
	\operatorname{tr}{\mathcal{A}}_n=0\,, \quad n \geq 0\,,
\label{eq:tr-An}
\end{equation*}
thus it gives us a parameterisation of the matrices ${\mathcal{A}}_n$, given by
\begin{equation*}
	{\mathcal{A}}_n=\frac{1}{A}
	\begin{bmatrix}
	 	l_{n} & \Theta_{n} \\
		-\Theta_{n-1}/\gamma_n &\;\	-l_{n}
	\end{bmatrix}\,.
\end{equation*}
Consequently, recalling the notation (\ref{eq:breve-An}), equation  (\ref{eq:aux-det}) reads as
\begin{equation}\label{eq:det-breveAn}
	\det \breve{\mathcal{A}}_n = \det\breve{\mathcal{A}}_{0}+ A\displaystyle\sum_{k=1}^{n} \frac{\Theta_{k-1}}{\gamma_{k}}\,,\quad n \geq 1\,,
\end{equation}
with $\det \breve{\mathcal{A}}_{0}=\displaystyle\frac{D}{w_0}\left(A+w_0B\right)-C^2/4.$

\subsection{Lax Pairs for Deformed Laguerre-Hahn Orthogonal Polynomials}
Let us now assume that the Stieltjes function depends on a deformation parameter $t$, say $f = f(x,t)$, and it satisfies a Riccati type differential equation
\begin{equation}
	\widehat{A} \partial_t f=\widehat{B}f^2+\widehat{C}f+\widehat{D}\,,
\label{eq:ric-dif-S-t}
\end{equation}
where $\widehat{A}, \widehat{B}, \widehat{C}, \widehat{D}$ are co-prime polynomials in $x$.
Proceeding analogously to Theorem \ref{theo:magnus} we get the results that follows.
\begin{theorem}\label{teo:sylv-t}
Let $f$ be a Stieltjes function given by (\ref{eq:expan-S}),
let $\{P_n\}_{n \geq 0}$ be the corresponding sequence of monic orthogonal polynomials,
and let $\beta_n, \gamma_n$ be its recurrence relation coefficients.  Assuming that $f=f(x,t)$ is such that it satisfies the Riccati-type differential equation (\ref{eq:ric-dif-S-t}),
then the matrices $Y_n$ given in (\ref{eq:sylvester-Yn-x}) satisfy the Sylvester equation
\begin{equation}
	\partial_t Y_n=\mathcal{B}_nY_n-Y_n\widehat{\mathcal{C}}\,, \quad n \geq 0\,,
\label{eq:sylvester-Yn-t}
\end{equation}
where
\begin{equation}
	\mathcal{B}_n = \frac{1}{\widehat{A}}
	\begin{bmatrix}
		 \widehat{l}_{n} & \widehat{\Theta}_{n} \\
		-\widehat{\Theta}_{n-1}/\gamma_n &\;\;
		\widehat{l}_{n-1}+(x-\beta_n)\widehat{\Theta}_{n-1}/\gamma_n
	\end{bmatrix}\,,
		\;\;\widehat{\mathcal{C}} = \frac{1}{\hat{A}}
	\begin{bmatrix}
		\widehat{C}/2 & -\widehat{D}/w_0 \\
		w_0\widehat{B} & -\widehat{C}/2+\partial_{t}(\ln w_0) \widehat{A}
	\end{bmatrix}
\end{equation}
with $\widehat{l}_n, \widehat{\Theta}_n$ polynomials of degrees given by
\begin{gather*}
	\deg(\widehat{l}_n)=\max\{\max\{\deg(\widehat{A})-1, \deg(\widehat{B})-1\}, \deg(\widehat{C})\}\,,
\\
	\deg(\widehat{\Theta}_n)=\max\{\max\{\deg(\widehat{A})-1, \deg(\widehat{B})-2\}, \deg(\widehat{C})-1\}\,,
\end{gather*}
satisfying the initial conditions
\begin{gather*}
	\frac{\widehat{\Theta}_{-1}}{w_0}=\widehat{D}\,, \widehat{\Theta}_0=-\frac{d}{dt}(\beta_0)\widehat{A}+(x-\beta_0)(\widehat{C}/2-\widehat{l}_0)+w_0\widehat{B}\,,\;
\\
	\widehat{l}_{-1}=\widehat{C}/2\,, \;\widehat{l}_0=-\widehat{C}/2-(x-\beta_0)\widehat{D}/w_0+\partial_{t}(\ln w_0) \widehat{A}\,.
\label{eq:initialcond-t}
\end{gather*}
\end{theorem}

The matrices $\mathcal{B}_n$ satisfy further identities, given in the next corollary.
\begin{corollary}
The following statements hold:\\
(a)\; the transfer matrix
  $\mathcal{M}_n$ from  (\ref{eq:3trr-mat}) satisfies the equation
\begin{equation}
	\partial_t \mathcal{M}_n=\mathcal{B}_n\mathcal{M}_n-\mathcal{M}_n{\mathcal{B}}_{n-1}\,, \quad n \geq 1\,;
\label{eq:Bnsylv-t}
\end{equation}
(b)\; the matrix
$ \breve{\mathcal{B}}_n:=
	\begin{bmatrix}
		 \widehat{l}_{n} & \widehat{\Theta}_{n} \\
		-\widehat{\Theta}_{n-1}/\gamma_n &\;\;
		 \widehat{l}_{n-1}+(x-\beta_n)\widehat{\Theta}_{n-1}/\gamma_n
	\end{bmatrix} $
satisfies
\begin{gather*}
	\operatorname{tr} \breve{\mathcal{B}}_n = \widehat{A}\,\sum_{k=0}^{n}\partial_{t}\ln(\gamma_k)\,, \quad n \geq 0\,,\label{eq:trBn-t}
\\
	\det \breve{\mathcal{B}}_n = \det \breve{\mathcal{B}}_{0}
	+ \widehat{A}\,\displaystyle\sum_{k=1}^{n}\left( \partial_{t}\ln(\gamma_k)\widehat{l}_{k-1} - \partial_{t}(\beta_k)\frac{\widehat{\Theta}_{k-1}}{\gamma_{k}} \right)\,,
	\quad n \geq 1\,,\;
\label{eq:detBn-t}
\end{gather*}
with
$\det \breve{\mathcal{B}}_0 = \left(-\frac{\displaystyle 1}{\displaystyle w_0}\partial_{t}\beta_0\widehat{A}+\widehat{B}\right)\displaystyle\widehat{D}
	+ \tfrac{1}{2}\partial_{t}(\ln w_0)\widehat{A}\widehat{C}
	- \tfrac{1}{4}\left(\widehat{C}\right)^2\,.$
\end{corollary}
\begin{proof}
Equation (\ref{eq:Bnsylv-t}) follows from the Lax pair
\begin{equation*}
	\left\{
	\begin{array}{l}
		{Y}_n = \mathcal{M}_n Y_{n-1}\\
		\partial_t{Y}_n = \mathcal{B}_nY_n-Y_n\widehat{\mathcal{C}}\, , \quad n\geq 1 \,.
	\end{array}\right.
\label{eq:lax-pair1}
\end{equation*}
The equations for the trace and determinant of $ \breve{\mathcal{B}}_n$ are deduced using similar techniques to \cite[Cor. 2]{AB-MNR-deformed}.
\end{proof}

\subsubsection{The Compatibility Conditions }
Here we would  like to stress that, in the Laguerre-Hahn setting, the compatibility conditions (indeed, the Schlesinger equations in (\ref{eq:schle-type}))  follow from the consistency
of the two linear differential equations for $Y_n$ combined with the consistency of the two
differential equations for the Stieltjes function (we refer the interested reader to \cite{AB-MNR-deformed}, where a different approach was used to deduce similar yet more restrictive conditions). The complete result is given as follows.
\begin{theorem}\label{teo:compat}
Let the systems (\ref{eq:sylvester-Yn-x}) and (\ref{eq:sylvester-Yn-t}) hold, with the corresponding Stieltjes function not algebraic.
Then, the following compatibility conditions hold, for all $n \geq 1$:
\begin{equation}
	\partial_t \mathcal{A}_n =\partial_x{\mathcal{B}}_n+{\mathcal{B}}_n{\mathcal{A}_n}-{\mathcal{A}_n}{\mathcal{B}}_n\,.
\label{eq:schle-type}
\end{equation}
\end{theorem}
\begin{proof}
Taking derivatives in (\ref{eq:sylvester-Yn-t}) with respect to $x$ and derivatives in (\ref{eq:sylvester-Yn-x}) with respect to $t$ we get, on account of the consistency
$\partial_t\partial_xY_n=\partial_x\partial_tY_n\,,$
\begin{equation}
\left( \partial_t\mathcal{A}_n-\partial_x{{\mathcal{B}}_n}+{\mathcal{A}_n}{\mathcal{B}}_n-{\mathcal{B}}_n{\mathcal{A}_n} \right)Y_n
 = Y_n\left( \partial_t{\mathcal{C}}-\partial_x\widehat{\mathcal{C}}+\mathcal{C}\widehat{\mathcal{C}}-\widehat{\mathcal{C}}\mathcal{C} \right)\,.
\label{eq:xt-Yn-compat}
\end{equation}
Let us now deduce
\begin{equation}
\partial_t{\mathcal{C}}-\partial_x\widehat{\mathcal{C}}+\mathcal{C}\widehat{\mathcal{C}}-\widehat{\mathcal{C}}\mathcal{C}=0_{2 \times 2}\,.
\label{eq:nullC}
\end{equation}
This follows from the compatibility of the differential equations (\ref{eq:ric-dif-S-t}) and (\ref{eq:ric-dif-f-x}). Indeed,
by taking derivatives in (\ref{eq:ric-dif-S-t}) with respect to $x$ and derivatives in (\ref{eq:ric-dif-f-x}) with respect to $t$ we get, on account of the consistency
$\partial_t\partial_xf=\partial_x\partial_tf\,,$
\begin{multline*}
	\partial_t\left(\frac{B}{A}\right)f^2+2\frac{B}{A}f\partial_tf+\partial_t\left(\frac{C}{A}\right)f+ \left(\frac{C}{A}\right)\partial_tf+ \partial_t\left(\frac{D}{A}\right)
\\
	= \partial_x\left(\frac{\widehat{B}}{\widehat{A}}\right)f^2+2\frac{\widehat{B}}{\widehat{A}}f\partial_xf+\partial_x
	\left(\frac{\widehat{C}}{\widehat{A}}\right)f+ \left(\frac{\widehat{C}}{\widehat{A}}\right)\partial_x f+ \partial_x\left(\frac{\widehat{D}}{\widehat{A}}\right).
\end{multline*}
The use of (\ref{eq:ric-dif-S-t}) and (\ref{eq:ric-dif-f-x}) in the previous equation yields, after simplifications,
\begin{multline*}
	\left[\partial_t\left(\frac{B}{A}\right)- \partial_x\left(\frac{\widehat{B}}{\widehat{A}}\right)+\frac{B\widehat{C}-\widehat{B}C}{A\widehat{A}}\right]f^2+
	\left[\partial_t\left(\frac{C}{A}\right)-\partial_x\left(\frac{\widehat{C}}{\widehat{A}}\right)+\frac{2(B\widehat{D}-\widehat{B}D)}{A\widehat{A}}\right]f
\\
	+ \partial_t\left(\frac{D}{A}\right)-\partial_x\left(\frac{\widehat{D}}{\widehat{A}}\right)+\frac{C\widehat{D}-\widehat{C}D}{A\widehat{A}} = 0\,.
\end{multline*}
As $f$ is not algebraic, then it must hold
\begin{align}
	\partial_t\left(\frac{B}{A}\right)-\partial_x\left(\frac{\widehat{B}}{\widehat{A}}\right)+\frac{B\widehat{C}-\widehat{B}C}{A\widehat{A}} &=0\,,
\label{eq:coefS2}\\
	\partial_t\left(\frac{C}{A}\right)-\partial_x\left(\frac{\widehat{C}}{\widehat{A}}\right)+\frac{2(B\widehat{D}-\widehat{B}D)}{A\widehat{A}} &=0\,,
\label{eq:coefS1}\\
	\partial_t\left(\frac{D}{A}\right)-\partial_x\left(\frac{\widehat{D}}{\widehat{A}}\right)+\frac{C\widehat{D}-\widehat{C}D}{A\widehat{A}} &=0\,.
\label{eq:coefS0}
\end{align}
Now, the $(2,1)$ entry of the matrix $\partial_t{\mathcal{C}}-\partial_x\widehat{\mathcal{C}}+\mathcal{C}\widehat{\mathcal{C}}-\widehat{\mathcal{C}}\mathcal{C}$ is zero by virtue of (\ref{eq:coefS2});
the  $(1,1)$ and $(2,2)$ entries are zero by virtue of (\ref{eq:coefS1}); the $(1,2)$ entry  is zero by virtue of (\ref{eq:coefS0}).
Hence, from (\ref{eq:xt-Yn-compat}), we are left with
\begin{equation*}
	\left( \partial_t\mathcal{A}_n-\partial_x{{\mathcal{B}}_n}+{\mathcal{A}_n}{\mathcal{B}}_n-{\mathcal{B}}_n{\mathcal{A}_n}\right)Y_n=0_{2 \times 2}\,.
\end{equation*}
The term $Y_{n}$ cancels out, as $\det(Y_{n}) \neq 0$ by virtue of (\ref{eq:liou}). 
Thus, we obtain (\ref{eq:schle-type}).
\end{proof}

\section{M\"{o}bius Transformation of the Stieltjes Function $f$ }\label{sec:3}

\subsection{The  System of Matrix Sylvester Equations}

We begin by specialising our weight to the {deformed} Jacobi weight defined in terms of parameters $\alpha, \beta, \mu$ by (\ref{eq:jacobi-w}),
\begin{equation*}
	w(x,t) = x^{\alpha}(1-x)^{\beta}(x- t)^{\mu}\,,
\end{equation*}
defined on a support that joins $0, 1, t$ (see \cite[Sec. 5]{magnus-jcam}). It satisfies the differential equation
\begin{equation*}
	A \partial_x w = C w \,,
\label{eq:der-w-x}
\end{equation*}
with the polynomials $A, C$ given by
\begin{equation}
	A(x)=x(x-1)(x-t)\,, \;\; C(x)=\alpha(x-1)(x-t)+\beta x(x-t)+\mu x(x-1)\,.
\label{eq:A-C-der-w-x}
\end{equation}
Let us consider the sequence of monic orthogonal polynomials,  $\{P_n\}_{n\geq 0}$, corresponding to $w$, as well as their associated polynomials,  $\{P_n^{(1)}\}_{n\geq 0}$.
According to the classification in \cite{maroni}, the weight $w$ and  $\{P_n\}_{n\geq 0}$
are semi-classical of class one as $\deg(A) = 3,\, \deg(C) = 2$.

Note that  $w$ also satisfies the differential equation
\begin{equation*}
	(x-t)\partial_{t}w = -\mu w \,.
\label{eq:der-w-t}
\end{equation*}
Therefore, by integrating with respect to $x$, we have the following relation for the first moments,
\begin{equation*}
	\partial_{t}w_1 = -\mu w_0+t\partial_{t}w_0 \,.
\label{eq:rel-w0-w1}
\end{equation*}
Using the general relation $w_1 = \beta_0 w_0$ in the equation above,
where $\beta_0=-P_1(0)$,  we  obtain
\begin{equation}
	(t-\beta_0)\partial_{t}\ln(w_0) = \partial_{t}\beta_0 + \mu \,.
\label{eq:der-w0}
\end{equation}
The equivalence between an homogeneous differential equation for the weight and the non-homogeneous linear first order differential equation for the corresponding
Stieltjes function is well-known \cite{maroni,shohat}.
The precise differential equations for the Stieltjes function of the weight (\ref{eq:jacobi-w}) are given in the following lemma.

\begin{lemma}
The Stieltjes function of the {deformed} Jacobi weight (\ref{eq:jacobi-w}), henceforth denoted by $f$, satisfies the inhomogeneous differential equations
\begin{align}
	x(x-1)(x-t)\partial_x f&=Cf+D\,,
\label{eq:der-f-x}\\
	(x-t)\partial_t f&=-\mu f+\hat{D}_0\,,
\label{eq:der-f-t}
\end{align}
with $C$ the polynomial given in (\ref{eq:A-C-der-w-x}) and $D, \hat{D}_0$ the polynomials given as follows: $D(x)=d_1x+d_0$, with
\begin{equation*}
	d_1=-(1+\alpha+\beta+\mu)w_0\,, d_0=\left[-(2+\alpha+\beta+\mu)\beta_0+1+\alpha+\beta+\mu+(1+\alpha+\beta)t\right]w_0\,,
\label{eq:d1-d0}
\end{equation*}
and
\begin{equation*}
	\hat{D}_0=\partial_{t}w_0\,.
\label{eq:hat-d}
\end{equation*}
\end{lemma}
\begin{proof}
The polynomials $D, \hat{D}_0$ are obtained by substituting the asymptotic expansion (\ref{eq:expan-S}) of $f$ into the differential equations (\ref{eq:der-f-x}) and (\ref{eq:der-f-t}),
respectively .
\end{proof}
\begin{remark}
For further purposes, we note the following relation,
\begin{equation}
	D(t)+t(t-1)\hat{D}_0=0\,,
\label{eq:rel-D-hat-D}
\end{equation}
which follows from the compatibility of equations (\ref{eq:der-f-x})-(\ref{eq:der-f-t}): taking derivatives of (\ref{eq:der-f-x}) with
respect to $t$ and derivatives of (\ref{eq:der-f-t}) with respect to $x$, and using $\partial_t \partial_x f=\partial_x\partial_t f $.
\end{remark}

Following (\ref{eq:tilde-f}), we shall now take the function $\tilde{f}$ to be given by (\ref{eq:til-f}),
\begin{equation*}
	\tilde{f}(x)=x-\beta_0-\frac{w_0}{f(x)}\,.
\end{equation*}
The quantities $w_0$ and $\beta_0$, acting as parameters defining the new function $\tilde{f}$ via (\ref{eq:til-f}),
will appear in the fundamental formulae to be subsequently deduced.

The function $\tilde{f}$ will play a central role in the sequel,
as the corresponding sequence of orthogonal polynomials belongs to the non semi-classical, Laguerre-Hahn class.
Indeed, using standard computations combining (\ref{eq:til-f}) with (\ref{eq:der-f-x}), the Riccati differential equation for $\tilde{f}$ is straightforward,
\begin{equation}
\tilde{A}\partial_x \tilde{f}=\tilde{B}\tilde{f}^2+\tilde{C}\tilde{f}+\tilde{D}\,, \label{eq:ric-f1-x}
 \end{equation}
 with coefficients
\begin{align}
	\tilde{A}(x) &= x(x-1)(x-t)
\label{eq:eg-pol-A1}\\
	\tilde{B}(x) &= -(1+\alpha+\beta+\mu)x- (2+\alpha+\beta+\mu)\beta_0+1+\alpha+\mu+(1+\alpha+\beta)t\,,
\label{eq:eg-pol-B1}\\
	\tilde{C}(x) &= (2+\alpha+\beta+\mu)x^2+[-(2+\alpha+\mu)-(2+\alpha+\beta)t+2\beta_0]x
\nonumber\\
	& \hspace{9.2cm}-\alpha t+2\tilde{B}(0)\beta_0\,,
\\
	\tilde{D}(x) &= [(3+\alpha+\beta+\mu)\beta_0^2-(2+\alpha+\mu+(2+\alpha+\beta)t)\beta_0+(1+\alpha)t]x
\nonumber\\
	&\hspace{1.4cm} -(2+\alpha+\beta+\mu)\beta_0^3+(1+\alpha+\mu+(1+\alpha+\beta)t)\beta_0^2-\alpha t \beta_0\,.
\label{eq:eg-pol-D1}
\end{align}
In order to proceed with our theory of the deformed Laguerre-Hahn polynomials,
we need to identify the role of the deformation parameter within the spectral differential system,
 as the zero $x=t$ of the polynomial $\tilde{A}$ given above.

The Riccati equation in the variable $t$ exhibits a dependence on the parameters $\beta_0$ and $w_0$ from (\ref{eq:til-f}).
Indeed, we have
\begin{equation}
	\hat{A}\partial_{t} \tilde{f}=\hat{B}\tilde{f}^2+\hat{C}\tilde{f}+\hat{D}\,,
\label{eq:ric-f1-t}
\end{equation}
with  coefficients
\begin{align}
	\hat{A}(x) &= x-t\,,
\\
	\hat{B} &= \partial_{t}\ln w_0\,,
\\
	\hat{C}(x) &= -\partial_{t}(\ln w_0)x+(2\beta_0-t)\partial_{t}\ln(w_0)+\mu \,,
\\
	\hat{D} &= -(\beta_0-t)\partial_{t}\beta_0\,.
\end{align}

  The next step is to deduce the differential systems satisfied by the orthogonal polynomials corresponding to $\tilde{f}$.
In order to simplify the text, we will use the following notations: we denote by $\{\tilde{P}_n\}_{n\geq 0}$
the sequence of monic orthogonal polynomials corresponding to $\tilde{f}$, and we denote by $(\tilde{\beta}_n)_{n\geq 0}, (\tilde{\gamma}_n)_{n\geq 0}$ its
sequences of recurrence relation coefficients. Let $\tilde{Y}_n$ be defined by (\ref{eq:Yn-til}),
\begin{equation*}
	\tilde{Y}_n=\left[
	\begin{matrix}
	  \tilde{P}_{n+1} & \tilde{P}_n^{(1)} \\
	  \tilde{P}_n & \tilde{P}_{n-1}^{(1)} \\
	\end{matrix}\right]\,.
\label{eq:tildeYn}
\end{equation*}
Upon utilising equations (\ref{eq:ric-f1-x}) and (\ref{eq:ric-f1-t}) we develop the results of Theorems \ref{teo:sylv-x} and \ref{teo:sylv-t},
now applied to the Stieltjes function $\tilde{f}$ and the 
corresponding sequence $\{\tilde{Y}_n\}_{n \geq 0}$.
We have the matrix Sylvester equations given in the following lemma.

\begin{lemma}\label{lemma:data-systems}
Let $\tilde{f}$ be the Stieltjes function defined by (\ref{eq:til-f}), and let $\{\tilde{P}_n\}_{n\geq 0}$ be the corresponding  sequence of monic orthogonal polynomials.
Under the previous notations, the following differential systems hold, for all $n \geq 0$:
\begin{align}
	\partial_x {\tilde{Y}}_n &= {\tilde{\mathcal{A}}}_n\tilde{Y}_n-\tilde{Y}_n{\tilde{\mathcal{C}}}\,,
\label{eq:eg-sylvester-Yn-x}\\
	\partial_t \tilde{Y}_n &= \tilde{\mathcal{B}}_n\tilde{Y}_n-\tilde{Y}_n\hat{\mathcal{C}}\,,
\label{eq:eg-sylvester-Yn-t}
\end{align}
where
\begin{equation}
	{\tilde{\mathcal{A}}}_n=\frac{1}{\tilde{A}}
	\begin{bmatrix}
		 \tilde{l}_{n} & \tilde{{\Theta}}_{n} \\
		-\tilde{\Theta}_{n-1}/{\tilde{\gamma}}_n &\;\;
		-\tilde{l}_{n}
	\end{bmatrix},
		\;\;{\tilde{\mathcal{C}}}=\frac{1}{\tilde{A}}
	\begin{bmatrix}
		 \tilde{C}/2 & -\tilde{D}/\tilde{w}_0 \\
		\tilde{w}_0\tilde{B} & -\tilde{C}/2
	\end{bmatrix},
\label{eq:eg-mat-An}
\end{equation}
and
\begin{equation}
	\tilde{\mathcal{B}}_n=\frac{1}{\hat{A}}
	\begin{bmatrix}
		 \hat{l}_{n} & \hat{\Theta}_{n} \\
		-\hat{\Theta}_{n-1}/{\tilde{\gamma}}_n &\;\;
		\hat{l}_{n-1}+(x-{\tilde{\beta}}_n)\hat{\Theta}_{n-1}/{\tilde{\gamma}}_n
	\end{bmatrix},
\label{eq:eg-mat-Bn}
\end{equation}
\begin{equation}
	\hat{\mathcal{C}}=\frac{1}{\hat{A}}
	\begin{bmatrix}
		 \hat{C}/2 & -\hat{D}/\tilde{w}_0 \\
		\tilde{w}_0\hat{B} & -\hat{C}/2+\partial_{t}(\ln \tilde{w}_0)\hat{A}
	\end{bmatrix},
\label{eq:eg-mat-hatC1}
\end{equation}
with $\tilde{A}, \tilde{B}, \tilde{C}, \tilde{D}$ the polynomials in (\ref{eq:ric-f1-x}), $\hat{A}, \hat{B}, \hat{C}, \hat{D}$ the polynomials in (\ref{eq:ric-f1-t}),
and $\tilde{w}_0$ denoting the zero moment in the asymptotic expansion (\ref{eq:expan-S1}) of $\tilde{f}$.

The entries of $\tilde{\mathcal{A}}_n$ and $\tilde{\mathcal{B}}_n$ are polynomials in the variable $x$, and given by
\begin{equation*}
	\tilde{l}_n(x)=\tilde{l}_{n,2}x^2+\tilde{l}_{n,1}x+\tilde{l}_{n,0}\,, \;\;
	\tilde{\Theta}_{n}(x)=\tilde{\Theta}_{n,1}x+\tilde{\Theta}_{n,0}\,, \;\;
	\hat{l}_n(x)=\hat{l}_{n,1}x+\hat{l}_{n,0}\,, \;\;
	\hat{\Theta}_{n}(x)=\hat{\Theta}_{n}\,,
\end{equation*}
where, for all $n \geq 1,$
\begin{align}
	\tilde{l}_{n,2} =&(\nu_n-1)/2\,,
\label{eq:ln2}\\
	\tilde{l}_{n,1} =&\sum_{j=0}^{n}\tilde{\beta}_j-\tfrac{1}{2}(\nu_n-1)(1+t)+\tfrac{1}{2}(\beta+\mu t)+\beta_0\,,\label{eq:ln1}
\\
	\tilde{l}_{n,0} =&\sum_{j=0}^{n}(\tilde{\beta}^2_j-(1+t)\tilde{\beta}_j+2\tilde{\gamma}_j)-\tfrac{1}{2}(t+1)(\alpha+\mu t)
	+\tfrac{1}{2}(\nu_n-1)t+\tfrac{1}{2}(\alpha+\mu t^2)
\nonumber\\
	&\hspace{2.7cm}+\nu_n\tilde{\gamma}_{n+1}+\beta_0(\beta_0-1-\beta-t)-\tfrac{1}{2}(\beta+\mu)t-\tfrac{1}{2}(\alpha+\mu t^2)\,,
\label{eq:ln0}\\
	\tilde{\Theta}_{n,1} =&-\nu_n\tilde{\gamma}_{n+1}\,,
\label{eq:thetan1}\\
	\tilde{\Theta}_{n,0} =&-\vartheta_n(t)\tilde{\gamma}_{n+1}\,,
\label{eq:thetan0}\\
	\hat{l}_{n,1} =&-\tfrac{1}{2}\partial_{t}\ln w_0\,,
\label{eq:hat-ln1}\\
	\hat{l}_{n,0} =&\tfrac{1}{2}\left((2\beta_0-t)\partial_{t}\ln w_0+\mu\right)-\sum_{j=0}^{n}\partial_{t}\tilde{\beta}_j
\label{eq:hat-ln0}\\
	{\hat{\Theta}_{n}} =&\left(\partial_{t}\ln w_0+\sum_{j=0}^{n+1}\partial_{t}\ln \tilde{\gamma}_j\right)\tilde{\gamma}_{n+1}\,,
\label{eq:hat-thetan}
\end{align}
with
\begin{align}
	\nu_n& = 2n+5+\alpha+\beta+\mu\,,
\label{eq:nu}\\
	\vartheta_n(t)& = \nu_n(\tilde{\beta}_{n+1}-1-t)+2\sum_{j=0}^{n}\tilde{\beta}_j+\tilde{\beta}_{n+1}+\mu t+\beta+2\beta_0\,.
\label{eq:vartheta}
\end{align}
The following initial conditions hold:
\begin{gather*}
	\frac{\tilde{\Theta}_{-1}}{\tilde{w}_0}=\tilde{D}\,,
	\tilde{\Theta}_0=\tilde{A}+(x-\tilde{\beta}_0)(\tilde{C}/2-\tilde{l}_0)+\tilde{w}_0\tilde{B}\,,
\\
	\tilde{l}_{-1}=\tilde{C}/2\,, \;\tilde{l}_0=-\tilde{C}/2-(x-\tilde{\beta}_0)\tilde{D}/\tilde{w}_0\,,
\\
	\frac{\hat{\Theta}_{-1}}{\tilde{w}_0}=\hat{D}\,, \hat{\Theta}_0=-\partial_{t}(\tilde{\beta}_0)\hat{A}+(x-\tilde{\beta}_0)(\hat{C}/2-\hat{l}_0)+\tilde{w}_0\hat{B}\,,\;
\\
	\hat{l}_{-1}=\hat{C}/2\,, \;\hat{l}_0=-\hat{C}/2-(x-\tilde{\beta}_0)\hat{D}/\tilde{w}_0+\partial_{t}(\ln \tilde{w}_0) \hat{A}\,.
\label{eq:initialcond-t}
\end{gather*}
\end{lemma}
\begin{proof}
The coefficients (\ref{eq:ln2})-(\ref{eq:hat-thetan}) follow from the differential equations in (\ref{eq:eg-sylvester-Yn-x}) and (\ref{eq:eg-sylvester-Yn-t}), using the expansions
\begin{eqnarray*}
	\tilde{P}_{n+1}(x)& =& x^{n+1}-(\tilde{\eta}_n+\tilde{\beta}_0)x^n+(\tilde{\kappa}_n+\tilde{\beta}_0\tilde{\eta}_n-\tilde{\gamma}_1)x^{n-1}+\dots+P_{n+1}(0)\,, \phantom{00}
\label{eq:exp-Pn} \\
	\tilde{P}_n^{(1)}(x)& =& x^n-\tilde{\eta}_nx^{n-1}+\tilde{\kappa}_nx^{n-2}+\dots+P_n^{(1)}(0)\,,
\label{eq:exp-Pn-1}
\end{eqnarray*}
where
\begin{equation*}
	\tilde{\eta}_n=\sum_{k=1}^{n}\tilde{\beta}_k\,,\;\;
	\tilde{\kappa}_n=\sum_{1\leq i < j\leq n}^n\tilde{\beta}_i\tilde{\beta}_j-\sum_{k=2}^{n}\tilde{\gamma}_k\,.
\label{eq:eta-nu}
\end{equation*}
\end{proof}

The data from $\hat{l}_n, \hat{\Theta}_n$ will be a fundamental tool to deduce the Toda equations
in Theorem \ref{teo:todaeqs} and to deduce the derivatives in Lemma \ref{lemma:derivatives}.
We now give further relations to be used later on.
\begin{corollary}
Under the previous notations, we have the following equation for $\hat{\Theta}_n$:
\begin{equation}
	\hat{\Theta}_n=
	-\sum_{j=1}^{n}\partial_{t}\tilde{\gamma}_j
	+t\sum_{j=1}^{n}\partial_{t}\tilde{\beta}_j
	-\tfrac{1}{2}\sum_{j=1}^{n}\partial_{t}\tilde{\beta}^2_j
	+(t-\tilde{\beta}_0)\partial_{t}\tilde{\beta}_0
	+\tilde{w}_0\partial_{t}\ln w_0\,.
\label{eq:another-hat-thetan}
\end{equation}
\end{corollary}
\begin{proof}The identity for $\hat{\Theta}_n$ follows from the coefficient of $x^n$ in the equation from the $(1,1)$ entry in (\ref{eq:eg-sylvester-Yn-t}), simplified using standard computations.
\end{proof}

\subsection{Fundamental Formulae and Derivatives of Related Quantities}
In this section we give relations satisfied by the quantities $\tilde{l}_{n,1}, \tilde{l}_{n,0}$ and  $\tilde{\Theta}_{n,0}$ that define the polynomial entries of the matrix
$\tilde{\mathcal{A}}_n$ for the system (\ref{eq:eg-sylvester-Yn-x}). Such relations will be a fundamental tool in the derivation of the sixth Painlev\'e equation in Section \ref{sec:3.3}.
Note that the poles of the matrix $\tilde{\mathcal{A}}_n$ are the zeroes of the polynomial $\tilde{A}$ given in (\ref{eq:eg-pol-A1}).

\subsubsection{Toda-type Equations for the Recurrence Coefficients}

We start by obtaining relations between the polynomial entries of the matrix of the systems (\ref{eq:eg-sylvester-Yn-x}) and (\ref{eq:eg-sylvester-Yn-t}).
The first result shows how $\hat{l}_n, \hat{\Theta}_n$ and $\tilde{l}_n, \tilde{\Theta}_n$ are related when evaluated at the pole $x=t$.

\begin{lemma}
Under the previous notations, the polynomial entries in the matrices (\ref{eq:eg-mat-An}) and  (\ref{eq:eg-mat-Bn}) satisfy
\begin{equation}
	\hat{l}_n(t)=-\frac{\tilde{l}_n(t)}{t(t-1)}\,, \;\;\; \hat{\Theta}_n=-\frac{\tilde{\Theta}_n(t)}{t(t-1)}\,.
\label{eq:rel-tilde-hat-pol}
\end{equation}
\end{lemma}
\begin{proof}
The relation (\ref{eq:rel-tilde-hat-pol}) is a direct consequence of the following equality for the residue matrices at $x=t$:
\begin{equation}
	\left. \mathcal{R}es\; \tilde{\mathcal{B}}_n\right|_{x=t}=-\left. \mathcal{R}es\; \tilde{\mathcal{A}}_n\right|_{x=t}.
\label{eq:iquality-res}
\end{equation}
Equation (\ref{eq:iquality-res}) is deduced as follows.
Firstly, recall that $\tilde{\mathcal{B}}_n$ has one single pole at $x=t$.
To compute the residue of $\tilde{\mathcal{B}}_n$ at $x=t,$ note that from (\ref{eq:eg-sylvester-Yn-t}) we have
\begin{equation*}
	\tilde{\mathcal{B}}_n=\partial_t(\tilde{Y}_n)\tilde{Y}_n^{-1}+(\tilde{Y}_n)\hat{\mathcal{C}}\tilde{Y}_n^{-1}\,.
\end{equation*}
Noting that $\partial_t(\tilde{Y}_n)\tilde{Y}_n^{-1}$ is a matrix of polynomial entries, thus entire, then the residue of $\tilde{\mathcal{B}}_n$ is given by
\begin{equation*}
	\left. \mathcal{R}es\; \tilde{\mathcal{B}}_n\right|_{x=t}=\left. \mathcal{R}es\; \tilde{Y}_n\hat{\mathcal{C}}\tilde{Y}_n^{-1}\right|_{x=t}\,.
\end{equation*}
As $\tilde{Y}_n\hat{\mathcal{C}}\tilde{Y}_n^{-1}$ is a matrix of rational entries with a single pole at $x=t$, then
\begin{equation*}
	\left. \mathcal{R}es\; \tilde{Y}_n\hat{\mathcal{C}}\tilde{Y}_n^{-1}\right|_{x=t}=\tilde{Y}_n(t)\hat{\hat{\mathcal{C}}}(t)\tilde{Y}_n^{-1}(t)\,,\quad \hat{\hat{\mathcal{C}}}
	=	\begin{bmatrix}
			\hat{C}/2 & -\hat{D}/\tilde{w}_0 \\
			\tilde{w}_0\hat{B} & -\hat{C}/2+\partial_{t}(\ln \tilde{w}_0)\hat{A}
		\end{bmatrix}.
\end{equation*}
For matters of simplicity, we only give details for the computations concerning the $(1,1)$ entry.
We have, in the account of $\det \tilde{Y}_n=-\prod_{k=}^{n}\tilde{\gamma}_k$,
\begin{multline}
	\left[\tilde{Y}_n(t)\hat{\hat{\mathcal{C}}}(t)\tilde{Y}_n^{-1}(t)\right]_{(1,1)}
	 = -\frac{1}{\prod_{k=}^{n}\tilde{\gamma}_k}\left\{\left(\tilde{P}_{n+1}(t)\tilde{P}^{(1)}_{n-1}(t)+\tilde{P}_{n}(t)\tilde{P}^{(1)}_{n}(t)\right)\tfrac{1}{2}\hat{C}(t)
	 \right.
\\
	\left.
	+ \tilde{P}_{n+1}(t)\tilde{P}_{n}(t)\frac{\hat{D}}{\tilde{w}_0}+\tilde{w}_0\tilde{P}_{n}(t)\tilde{P}^{(1)}_{n-1}(t)\hat{B}\right\}\,.
\label{eq:aux1-res}
\end{multline}
We now compare (\ref{eq:aux1-res}) with the residue at $x=t$ of the $(1,1)$ entry in $\tilde{\mathcal{A}}_n$. Note that
\begin{equation*}
	\left[\left. \mathcal{R}es\; \tilde{\mathcal{A}}_n\right|_{x=t}\right]_{(1,1)}=\frac{\tilde{l}_n(t)}{t(t-1)}\,.
\end{equation*}
The computation of $\tilde{l}_n$, in accordance with the definition in (\ref{eq:ln-def}), yields
\begin{multline}
	\tilde{l}_n(t)=-\frac{1}{\prod_{k=}^{n}\tilde{\gamma}_k}\left\{\left(\tilde{P}_{n+1}(t)\tilde{P}^{(1)}_{n-1}(t)
	+\tilde{P}_{n}(t)\tilde{P}^{(1)}_{n}(t)\right)\tfrac{1}{2}\tilde{C}(t)
	\right.
\\
	\left.
	+ \tilde{P}_{n+1}(t)\tilde{P}_{n}(t)\frac{\tilde{D}(t)}{w_0^{(1)}}+w_0^{(1)}\tilde{P}_{n}(t)\tilde{P}^{(1)}_{n-1}(t)\tilde{B}(t)\right\}\,.
\label{eq:aux2-res}
\end{multline}
Taking into account the data contained in $\tilde{B}, \tilde{C}, \tilde{D}, \hat{B}, \hat{C}$ and $\hat{D},$ from equations (\ref{eq:aux1-res}) and (\ref{eq:aux2-res}), we get
\begin{multline*}
	\left[\left. \mathcal{R}es\; \tilde{\mathcal{A}}_n\right|_{x=t}\right]_{(1,1)}+\left[\left. \mathcal{R}es\; \tilde{\mathcal{B}}_n\right|_{x=t}\right]_{(1,1)}
\\
	=-\frac{1}{w_0\prod_{k=}^{n}\tilde{\gamma}_k}\left\{-(t-\beta_0)\left(\tilde{P}_{n+1}(t)\tilde{P}^{(1)}_{n-1}(t)+\tilde{P}_{n}(t)\tilde{P}^{(1)}_{n}(t)\right)
	\right.
\\
	\left.
	+\frac{(t-\beta_0)^2}{\tilde{w}_0} \tilde{P}_{n+1}(t)\tilde{P}_{n}(t)+\tilde{w}_0\tilde{P}_{n}(t)\tilde{P}^{(1)}_{n-1}(t)\right\}\left(\hat{D}_0+\frac{D(t)}{t(t-1)}\right)\,,
\end{multline*}
with $D, \hat{D}_0$ the polynomials given in (\ref{eq:der-f-x}) and (\ref{eq:der-f-t}), respectively.
As $\hat{D}_0+\frac{D(t)}{t(t-1)}=0$ (recalling (\ref{eq:rel-D-hat-D})), then we conclude that
\begin{equation*}
	\left[\left. \mathcal{R}es\; \tilde{\mathcal{A}}_n\right|_{x=t}\right]_{(1,1)}+\left[\left. \mathcal{R}es\; \tilde{\mathcal{B}}_n\right|_{x=t}\right]_{(1,1)} = 0.
\end{equation*}
Similar computations lead to
\begin{equation*}
	\left[\left. \mathcal{R}es\; \tilde{\mathcal{A}}_n\right|_{x=t}\right]_{(i,j)}+\left[\left. \mathcal{R}es\; \tilde{\mathcal{B}}_n\right|_{x=t}\right]_{(i,j)} = 0, \quad i,j =1, 2\,,
\end{equation*}
hence (\ref{eq:iquality-res}) follows.
\end{proof}

As a consequence, we can show that the derivatives of the recurrence coefficients are given in terms of the functions
$\tilde{\Theta}_n, \tilde{l}_n$ via the following Toda-type equations.
\begin{theorem}\label{teo:todaeqs}
Under the conditions and notations of Lemma \ref{lemma:data-systems}, the recurrence relation coefficients  $(\tilde{\beta}_n)_{n\geq 0}, (\tilde{\gamma}_n)_{n\geq 0}$
satisfy the following Toda-type equations, for all $n \geq 1$ :
\begin{align}
	t(t-1)\frac{d}{dt}\tilde{\beta}_n&=\tilde{\beta_n}(\tilde{\beta}_n-1)+(\nu_{n-1}+2)\tilde{\gamma}_{n+1}-(\nu_{n-1}-2)\tilde{\gamma}_n\,,\label{eq:toda-beta}
\\
	t(t-1)\frac{d}{dt}\ln \tilde{\gamma}_n&=-2+(\nu_{n-1}+1)\tilde{\beta}_{n}-(\nu_{n-1}-3)\tilde{\beta}_{n-1}\,.\label{eq:toda-gamma}
\end{align}
\end{theorem}
\begin{proof}
Because of (\ref{eq:rel-tilde-hat-pol}), equations (\ref{eq:toda-beta})-(\ref{eq:toda-gamma}) are obtained, respectively, from
\begin{align*}
	t(t-1)\left(\hat{l}_n(t)-\hat{l}_{n-1}(t)\right)&=\tilde{l}_{n-1}(t)-\tilde{l}_{n}(t)\,,
\\
	t(t-1)\left(\frac{\hat{\Theta}_n}{\tilde{\gamma_{n+1}}}-\frac{\hat{\Theta}_{n-1}}{\tilde{\gamma_{n}}}\right)&=\frac{\tilde{\Theta}_{n-1}}{\tilde{\gamma}_n}(t)-
	\frac{\tilde{\Theta}_{n}}{\tilde{\gamma}_{n+1}}(t)\,,
\end{align*}
by using the data given in Lemma \ref{lemma:data-systems}.
\end{proof}

\subsubsection{The Schlesinger-type Equations}
We shall now analyse the compatibility conditions that result from the consistency of the two systems for $\tilde{Y}_n$, (\ref{eq:eg-sylvester-Yn-x}) and (\ref{eq:eg-sylvester-Yn-t}).
The compatibility conditions stated in Theorem \ref{teo:compat} are now given by
\begin{equation}
	\partial_t \tilde{\mathcal{A}}_n =\partial_x\tilde{\mathcal{B}}_n+\tilde{\mathcal{B}}_n\tilde{\mathcal{A}}_n-\tilde{\mathcal{A}}_n\tilde{\mathcal{B}}_n\,.
\label{eq:eg-compat}
\end{equation}
Denoting the poles of $\tilde{\mathcal{A}}_n$  in (\ref{eq:eg-mat-An}) by $x_1=0, x_2=1, x_3=t,$ then we may write
\begin{equation}
\tilde{\mathcal{A}}_n=\sum_{j=1}^{3}\frac{\tilde{\mathcal{A}}_{n,j}}{x-x_j}\,,\label{eq:frac-An}
\end{equation}
where the residue matrices are
\begin{equation*}
	{\tilde{\mathcal{A}}}_{n,j}=\frac{1}{\partial_x \tilde{A}(x_j)}
	\begin{bmatrix}
		 \tilde{l}_{n}(x_j) & \tilde{{\Theta}}_{n}(x_j) \\
		-\frac{\tilde{\Theta}_{n-1}}{{\tilde{\gamma}}_n}(x_j) &\;\;
		-\tilde{l}_{n}(x_j)
	\end{bmatrix}\,, \quad j=1, 2, 3\,.
\end{equation*}
The matrix $\tilde{\mathcal{B}}_n$, as implied by (\ref{eq:eg-mat-Bn})  and (\ref{eq:iquality-res}), has the form
\begin{equation}
	\tilde{\mathcal{B}}_n=\tilde{\mathcal{B}}_{n,\infty}- \frac{\tilde{\mathcal{A}}_{n,3}}{x-t}\,,
\label{eq:frac-Bn-hat}
\end{equation}
with residue matrix
\begin{equation*}
	\tilde{\mathcal{B}}_{n,\infty}=
	\begin{bmatrix}
		 \hat{l}_{n,1} & 0 \\
		 0 &\;\;\hat{l}_{n-1,1}+\frac{\hat{\Theta}_{n-1}}{\tilde{\gamma}_n}
	\end{bmatrix}.
\end{equation*}

Due to the forms of (\ref{eq:frac-An}) and (\ref{eq:frac-Bn-hat}), then (\ref{eq:eg-compat}) becomes a Schlesinger-type equation
\begin{equation}
\sum_{j=1}^{3}\frac{\partial_t\tilde{\mathcal{A}}_{n,j}}{x-x_j}=\sum_{j=1}^{3}\frac{[\hat{\mathcal{\mathcal{B}}}_{n\infty}\,,\,\tilde{\mathcal{A}}_{n,j}]}{x-x_j}+
\frac{[{\mathcal{\mathcal{A}}}_{n,1}\,,\,\tilde{\mathcal{A}}_{n,3}]}{(x-x_1)(x-x_3)}+\frac{[{\mathcal{\mathcal{A}}}_{n,2}\,,\,\tilde{\mathcal{A}}_{n,3}]}{(x-x_2)(x-x_3)}\,.\label{eq:eg-schle}
\end{equation}

The above equation (\ref{eq:eg-schle}) will now be used to obtain the derivatives of the quantities $\tilde{l}_{n,1}, \tilde{l}_{n,0}$ and $\Theta_{n,0}/\tilde{\gamma}_{n+1}$
that parameterise the matrix  $\tilde{\mathcal{A}}_n$. 
\begin{lemma}\label{lemma:derivatives}
Under the previous notations, the following equations hold, for all $n \geq 1$:
\begin{align}
	\partial_{t}\tilde{l}_{n,0}&=\frac{\tilde{l}_{n,0}}{t}-\frac{\tilde{\gamma}_{n+1}(\nu_n\vartheta_{n-1}-\vartheta_n\nu_{n-1})}{t(t-1)}\,,
\label{eq:der-ln0}\\
	\partial_{t}\tilde{l}_{n,1}&=\frac{\nu_n -1}{2(t-1)}+\frac{\tilde{l}_{n,1}}{t-1}+\frac{\tilde{l}_{n,0}}{t(t-1)}\,,
\label{eq:der-ln1}\\
	\partial_{t}\vartheta_n&=\frac{-\vartheta_n-\vartheta_n^2+2\xi_n}{t(t-1)}\,,
\label{eq:der-vartheta}
\end{align}
with \begin{equation}
\xi_n=\vartheta_n\tilde{l}_{n,1}-\nu_n\tilde{l}_{n,0}\,.\label{eq:eg-xin}
\end{equation}
\end{lemma}
\begin{proof}
Equations (\ref{eq:der-ln0})-(\ref{eq:der-vartheta}) are obtained as follows:
(\ref{eq:der-ln0}) follows from the equation  in the  $(1,1)$ entry of  (\ref{eq:eg-schle}) at the residue $x = 0$;
 (\ref{eq:der-ln1}) follows from the equation  in the $(1,1)$ entry of (\ref{eq:eg-schle}) at the residue $x = 1$;
 (\ref{eq:der-vartheta}) follows from the equation in the $(1,2)$ entry of (\ref{eq:eg-schle}) at the residue $x = 0$
combined with the equation in the  $(1,2)$ entry at the residue $x = 1$ and where the data (\ref{eq:hat-thetan}) was used.
\end{proof}

\subsubsection{Bilinear Relations: the Determinant of the Matrix $\tilde{\mathcal{A}}_n$}
Let us analyse the equations arising from the determinant of the matrix $\tilde{\mathcal{A}}_n$ (see (\ref{eq:aux-det}) and (\ref{eq:det-breveAn})),
which in the present situation are written as
\begin{equation}
	-\tilde{l}_n^2+\tilde{\Theta}_{n}\frac{\tilde{\Theta}_{n-1}}{\tilde{\gamma}_{n}}=\varphi+\tilde{A}\sum_{k=1}^{n}\frac{\tilde{\Theta}_{k-1}}{\tilde{\gamma}_{k}}\,,
\label{eq:eg-det}
\end{equation}
with
\begin{equation}
	\varphi=\frac{\tilde{D}}{\tilde{w}_0}\left(\tilde{A}+\tilde{w}_0\tilde{B}\right)-\frac{\tilde{C}^2}{4}\,.
\label{eq:eg-det-xi}
\end{equation}
Evaluating (\ref{eq:eg-det}) at the zeroes of $\tilde{A}$ we have, respectively,
\begin{align}
	\frac{\tilde{\Theta}_{n-1}(0)}{\tilde{\gamma}_n}&=\frac{\tilde{l}^2_{n}(0)+\varphi(0)}{\tilde{\Theta}_{n}(0)}\,,\label{eq:eg-det-x=0}\\
	\frac{\tilde{\Theta}_{n-1}(1)}{\tilde{\gamma}_n}&=\frac{\tilde{l}^2_{n}(1)+\varphi(1)}{\tilde{\Theta}_{n}(1)}\,,\label{eq:eg-det-x=1}\\
	\frac{\tilde{\Theta}_{n-1}(t)}{\tilde{\gamma}_n}&=\frac{\tilde{l}^2_{n}(t)+\varphi(t)}{\tilde{\Theta}_{n}(t)}\,. \label{eq:eg-det-x=t}
\end{align}
Also, for subsequent purposes, by utilising (\ref{eq:thetan1}) and (\ref{eq:thetan0}), we shall write (\ref{eq:eg-det-x=0}) and (\ref{eq:eg-det-x=1}) respectively, as
\begin{align}
	\tilde{\gamma}_{n+1}\vartheta_{n-1} &= \frac{\tilde{l}^2_{n,0}+\varphi(0)}{\vartheta_n}\,,
\label{eq:auxaux-det-x=0}\\
	\tilde{\gamma}_{n+1}\nu_{n-1} &=
	\frac{\tilde{l}^2_{n}(1)+\varphi(1)}{\nu_n+\vartheta_n}-\tilde{\gamma}_{n+1}\vartheta_{n-1}\,.
\label{eq:auxaux-det-x=1}
\end{align}

\begin{lemma}
Under the previous notations, the following equation holds, for all $n \geq 1$:
\begin{multline}
	\tilde{l}_{n,0}=-\frac{\varphi(0)}{t(\nu_n-1)}+\frac{\vartheta_n\varphi(1)}{(t-1)(\nu_n-1)(\nu_n+\vartheta_n)}-\frac{\vartheta_n\varphi(t)}{t(t-1)(\nu_n-1)(\nu_nt+\vartheta_n)}
\\
	-\frac{\left[(\nu_n-1)\vartheta_n[t(t+1)\nu_n+(t^2+t+1)\vartheta_n]/4+[t\nu_n+(t+1)\vartheta_n]\xi_n+\xi_n^2/(\nu_n-1)\right]}{(\nu_n+\vartheta_n)(\nu_n t+\vartheta_n)}\, ,
\label{eq:eg-ident-ln0}
\end{multline}
with $\xi_n$ given by (\ref{eq:eg-xin}).
\end{lemma}
\begin{proof} As the
polynomial $\tilde{\Theta}_n$ has degree one, it satisfies the following identity, for all $n \geq 1$:
\begin{equation*}
\frac{\tilde{\Theta}_{n-1}(0)}{t}+\frac{\tilde{\Theta}_{n-1}(1)}{1-t}+\frac{\tilde{\Theta}_{n-1}(t)}{t(t-1)}=0\,. \label{eq:ln0-aux1}
\end{equation*}
Thus, from (\ref{eq:eg-det-x=0})-(\ref{eq:eg-det-x=t}) we have
\begin{equation*}
	\frac{\tilde{l}^2_{n}(0)+\varphi(0)}{t\tilde{\Theta}_{n}(0)}-\frac{\tilde{l}^2_{n}(1)+\varphi(0)}{(t-1)\tilde{\Theta}_{n}(1)}
	+\frac{\tilde{l}^2_{n}(t)+\varphi(t)}{t(t-1)\tilde{\Theta}_{n}(t)}=0\,.
\label{eq:ln0-aux2}
\end{equation*}
The previous equation can be written as follows:
\begin{multline}
	\tilde{l}_{n,2}^2\left[\frac{t^3\tilde{\Theta}_n(1)-\tilde{\Theta}_n(t)}{(t-1)\tilde{\Theta}_n(1)\tilde{\Theta}_n(t)}\right]+
	\tilde{l}_{n,1}^2\left[\frac{t\tilde{\Theta}_n(1)-\tilde{\Theta}_n(t)}{(t-1)\tilde{\Theta}_n(1)\tilde{\Theta}_n(t)}\right]
\\
	+\tilde{l}_{n,0}^2\left[\frac{(t-1)\tilde{\Theta}_n(1)\tilde{\Theta}_n(t)-t\tilde{\Theta}_n(0)\tilde{\Theta}_n(t)
	+\tilde{\Theta}_n(0)\tilde{\Theta}_n(1)}{t(t-1)\tilde{\Theta}_n(0)\tilde{\Theta}_n(1)\tilde{\Theta}_n(t)}\right]
\\
	+2\tilde{l}_{n,2}\tilde{l}_{n,1}\left[\frac{t^2\tilde{\Theta}_n(1)-\tilde{\Theta}_n(t)}{(t-1)\tilde{\Theta}_n(1)\tilde{\Theta}_n(t)}\right]
	+2\tilde{l}_{n,2}\tilde{l}_{n,0}\left[\frac{t\tilde{\Theta}_n(1)-\tilde{\Theta}_n(t)}{(t-1)\tilde{\Theta}_n(1)\tilde{\Theta}_n(t)}\right]
	+2\tilde{l}_{n,1}\tilde{l}_{n,0}\left[\frac{\tilde{\Theta}_n(1)-\tilde{\Theta}_n(t)}{(t-1)\tilde{\Theta}_n(1)\tilde{\Theta}_n(t)}\right]
\\
	+\frac{\varphi(0)}{t\tilde{\Theta}_n(0)}-\frac{\varphi(1)}{(t-1)\tilde{\Theta}_n(1)}+\frac{\varphi(t)}{t(t-1)\tilde{\Theta}_n(t)}=0\,.
\label{eq:ln-aux3}
\end{multline}
Using the identities
\begin{align*}
	\frac{\tilde{\Theta}_{n}(1)-\tilde{\Theta}_{n}(t)}{t-1}&=-\tilde{\Theta}_{n,1}\,,\\
	\frac{t\tilde{\Theta}_{n}(1)-\tilde{\Theta}_{n}(t)}{t-1}&=\tilde{\Theta}_{n}(0)\,,\\
	\frac{t^2\tilde{\Theta}_{n}(1)-\tilde{\Theta}_{n}(t)}{t-1}&=t\tilde{\Theta}_{n}(0)+\tilde{\Theta}_{n}(t)\,,\\
	\frac{t^3\tilde{\Theta}_{n}(1)-\tilde{\Theta}_{n}(t)}{t-1}&=t^2\tilde{\Theta}_{n}(0)+(t+1)\tilde{\Theta}_{n}(t)\,,\\
	\frac{(t-1)\tilde{\Theta}_{n}(1)\tilde{\Theta}_{n}(t)-t\tilde{\Theta}_{n}(0)\tilde{\Theta}_{n}(t)+\tilde{\Theta}_{n}(0)\tilde{\Theta}_{n}(1)}{t(t-1)}&=\tilde{\Theta}_{n,1}^2\,,
\end{align*}
then (\ref{eq:ln-aux3}) can be simplified as
\begin{multline}
	\frac{2\tilde{l}_{n,2}\tilde{l}_{n,0}}{\tilde{\Theta}_{n}(0)}
	+\tilde{l}_{n,2}^2\tilde{\Theta}_n(0)\frac{[t^2\tilde{\Theta}_{n}(0)+(t+1)\tilde{\Theta}_{n}(t)]}{\tilde{\Theta}_{n}(0)\tilde{\Theta}_{n}(1)\tilde{\Theta}_{n}(t)}
	+2\tilde{l}_{n,2}\xi_n\frac{[t\tilde{\Theta}_{n}(0)+\tilde{\Theta}_{n}(t)]}{\tilde{\Theta}_{n}(0)\tilde{\Theta}_{n}(1)\tilde{\Theta}_{n}(t)}
\\
	+\frac{\xi_n^2}{\tilde{\Theta}_{n}(0)\tilde{\Theta}_{n}(1)\tilde{\Theta}_{n}(t)}
	+\frac{\varphi(0)}{t\tilde{\Theta}_n(0)}
	-\frac{\varphi(1)}{(t-1)\tilde{\Theta}_n(1)}
	+\frac{\varphi(t)}{t(t-1)\tilde{\Theta}_n(t)}=0 \,,
\label{eq:aux-aux-ln0}
\end{multline}
with $\xi_n$ defined in (\ref{eq:eg-xin}).
Solving (\ref{eq:aux-aux-ln0}) for $\tilde{l}_{n,0}$ yields (\ref{eq:eg-ident-ln0}).
\end{proof}

\subsection{Derivation of P$_{\textrm{VI}}$ }\label{sec:3.3}
Now, once the derivatives of $\vartheta_n$ and $\xi_n$ have been computed, our remaining task is to deduce the Painlev\'e equation
and to obtain the transformation that identifies the Hamiltonian equations for the P$_{\textrm{VI}}$ system.
We now give the proof of Theorem \ref{teo:P6}. Recall that dash means derivative with respect to $t$.

\begin{proof}
Firstly, let us note that, in terms of $q$ defined in (\ref{eq:hamilt-q}),
 \begin{equation*}
	q=-\vartheta_n/\nu_n\,,
\end{equation*}
we have the following identities:
equation (\ref{eq:der-ln0}) is given, taking into account (\ref{eq:auxaux-det-x=0})-(\ref{eq:auxaux-det-x=1}), by
\begin{equation}
	\tilde{l}'_{n,0} = \frac{1}{t}\tilde{l}_{n,0}
	-\frac{1}{t(t-1)}\left[ (\tilde{l}^2_{n,0}+\varphi(0))\displaystyle\left(\frac{q-1}{q}\right)+(\tilde{l}^2_{n}(1)+\varphi(1))\displaystyle\left(\frac{q}{1-q}\right) \right] \,,
\label{eq:der-ln0-q}
\end{equation}
equation (\ref{eq:der-vartheta}) by
\begin{equation}
	t(t-1)q'=-q+\nu_nq^2-2\xi_n/\nu_n\,,
\label{eq:der-vartheta-q}
\end{equation}
and equation (\ref{eq:eg-ident-ln0}) by
\begin{multline}
	\tilde{l}_{n,0} = -\frac{\varphi(0)}{(\nu_n-1)t} - \frac{\varphi(1)}{(\nu_n-1)(t-1)}\left(\frac{q}{1-q}\right)+\frac{\varphi(t)}{(\nu_n-1)t(t-1)}\left(\frac{q}{t-q}\right)
\\
	-\frac{\left[-(\nu_n-1)q+[t(t+1)-(t^2+t+1)q]/4+[t-(t+1)q]\left(\xi_n/\nu_n\right)+\left(\xi_n/\nu_n\right)^2/(\nu_n-1)\right]}{(1-q)(t-q)}\,.
\label{eq:eg-ln0-q}
\end{multline}
Taking derivatives of $\xi_n/\nu_n=-q\tilde{l}_{n,1}-\tilde{l}_{n,0}$ (cf. (\ref{eq:eg-xin})) we have 
\begin{equation}
	\left(\frac{\xi_n}{\nu_n}\right)'=-q'\tilde{l}_{n,1}-q\tilde{l}'_{n,1}-\tilde{l}'_{n,0}\,.
\label{eq:dt-xin-over-nu}
\end{equation}
The use of (\ref{eq:der-ln1}) and (\ref{eq:der-ln0-q}) in (\ref{eq:dt-xin-over-nu}) gives us
\begin{multline*}
	\left(\frac{\xi_n}{\nu_n}\right)'
	= -q'\tilde{l}_{n,1}-q\frac{(\nu_n-1)}{2(t-1)}-q\frac{\tilde{l}_{n,1}}{t}-q\frac{\tilde{l}_{n,0}}{t(t-1)}-\frac{\tilde{l}_{n,0}}{t}
\\
	+\frac{\tilde{l}^2_{n,0}+\varphi(0)}{t(t-1)}\displaystyle\left(\frac{q-1}{q}\right) + \frac{\tilde{l}^2_{n}(1)+\varphi(1)}{t(t-1)}\left(\frac{q}{1-q}\right)\,.
\end{multline*}
Using $\tilde{l}_{n,1}=-\displaystyle\frac{\xi_n/\nu_n+\tilde{l}_{n,0}}{q}$ as well as (\ref{eq:der-vartheta-q}) and (\ref{eq:eg-ln0-q})  in the previous equation, we get
\begin{multline}
	\left(\frac{\xi_n}{\nu_n}\right)' =
	\left(\frac{1}{1-q}+\frac{1}{-q}+\frac{1}{t-q}\right)t(t-1)\left(\frac{q'}{2}\right)^2
	+\left(2\nu_nq-1-\frac{t(t-1)}{t-q}\right)\left(\frac{q'}{2}\right)
\\
	+ \frac{\varphi(1)}{t(t-1)^2}\frac{q(t-q)}{1-q}-\frac{\varphi(0)}{t^2(t-1)}\frac{(1-q)(t-q)}{q}
	-\left(\frac{1}{4}+\frac{\varphi(t)}{t^2(t-1)^2}\right)\frac{q(1-q)}{t-q}
\\
	-\frac{q(1-q)(t-q)}{4t(t-1)}\,.
\label{eq:der-xi-q}
\end{multline}
Let us now compute the second derivative of $q$. By differentiating (\ref{eq:der-vartheta-q}) we get
\begin{equation}
	t(t-1)q''=-2tq'+2\nu_nq q'-2\left(\frac{\xi_n}{\nu_n}\right)'\,.
\label{eq:2ndderq}
\end{equation}
Employing (\ref{eq:der-xi-q}) in (\ref{eq:2ndderq}) and taking into account that $\varphi$ as defined by (\ref{eq:eg-det-xi}) has the evaluations
\begin{equation}
	\;\;\frac{\varphi(0)}{t^2}=-\frac{\alpha^2}{4}\,,
	\;\;\frac{\varphi(1)}{(t-1)^2}=-\frac{\beta^2}{4}\,,
	\;\;\frac{\varphi(t)}{t^2(t-1)^2} = -\frac{\mu^2}{4}\,,\,
\label{eq:quant-varphi}
\end{equation}
we finally obtain
\begin{multline}
	q'' = \frac{1}{2}\left(\frac{1}{q}+\frac{1}{q-1}+\frac{1}{q-t}\right)\left({q'}\right)^2
	-\left(\frac{1}{t}+\frac{1}{t-1}+\frac{1}{q-t}\right)q'
\\
 	+\frac{q(q-1)(q-t)}{t^2(t-1)^2}\left[\frac{1}{2}-\frac{\alpha^2 t}{2q^2}+\frac{\beta^2(t-1)}{2(q-1)^2}+\frac{(1-\mu^2)t(t-1)}{2(q-t)^2}\right]\,.
\label{eq:P6-q}
\end{multline}
Thus, $q(t)$ satisfies (\ref{eq:P6}) with parameters
$(\delta_1, \delta_2, \delta_3, \delta_4)= (\frac{1}{2}, -\frac{1}{2}\alpha^2, \frac{1}{2}\beta^2, \frac{1}{2}(1-\mu^2))$.
Correspondingly, the parameters $v_k$ in the Hamiltonian (\ref{eq:H6}) are such that
\begin{equation*}
	 (v_1-v_2)^2=1\,, \;  (v_3+v_4)^2= \alpha^2\,,\; (v_3-v_4)^2=\beta^2\,, \; (1-v_1-v_2)^2=\mu^2\,,
\end{equation*}
which reads as (\ref{eq:parameter-varphi}).
From the first of the Hamiltonian equations 
(\ref{eq:hamilton-eqs}) combined with (\ref{eq:der-vartheta-q}) and taking into account that $\xi_n/\nu_n=-q\tilde{l}_{n,1}-\tilde{l}_{n,0}$,
we deduce the equation for $p(t)$ as given by (\ref{eq:2nd-hamilt-p}).
  \end{proof}

It remains to specify the initial conditions for the differential equation to uniquely characterize $q(t)$. This will be done in the next subsection.
\subsubsection{Boundary Values for the Laguerre-Hahn P$_{\textrm{VI}}$ Solution}

We first begin by establishing the initial conditions of the differential equation satisfied by $q$. 
\begin{lemma}
The function $q(t)$ defined by (\ref{eq:hamilt-q}), solution of the differential equation (\ref{eq:P6}) with parameters
$(\delta_1, \delta_2, \delta_3, \delta_4)= (\frac{1}{2}, -\frac{1}{2}\alpha^2, \frac{1}{2}\beta^2, \frac{1}{2}(1-\mu^2))$,
is unique via the specification of 
$\beta_0$ and $\tilde{\beta}_0$ in the equations
\begin{align}
	q(0)&=-\left((2+\alpha+\beta+\mu)\tilde{\beta}_0(0)+2\beta_0(0)-3-\alpha-\mu\right)/\nu_n\,,
\label{eq:ic-0}\\
	q(1)&=-\left((2+\alpha+\beta+\mu)\tilde{\beta}_0(1)+2\beta_0(1)-\nu_n-3-\alpha\right)/\nu_n\,,
\label{eq:ic-1}
\end{align}
as well as in
\begin{align}
	q'(0)&=-\left(-g_n(0)+(2+\alpha+\beta+\mu)\tilde{\beta}'_0(0)+2{\beta}'_0(0)-\nu_n+\mu\right)/\nu_n\,,
\label{eq:ic-vartheta0}\\
	q'(1)&=-\left(g_n(1)+(2+\alpha+\beta+\mu)\tilde{\beta}'_0(1)+2{\beta}'_0(1)-\nu_n+\mu\right)/\nu_n\,,
\label{eq:ic-vartheta1}
\end{align}
where $g_n$ is defined in terms of the moments in the asymptotic expansion (\ref{eq:expan-S1}) by
\begin{equation}
	g_n(t)=\left(\ln \left(\tilde{\Delta}_{n+2}/\tilde{\Delta}_{n+1}\right)-\ln \tilde{w}_{0}\right)'\,,
\label{eq:gn}
\end{equation}
where $\tilde{\Delta}_n:=\det[\tilde{w}_{j+k-2}]_{j, k=1,\dots, n}\,,$ with $\tilde{w}_n$ defined in terms of the moments of the weight (\ref{eq:jacobi-w}) via relation (\ref{eq:moments}).
\end{lemma}
\begin{proof}
Taking $t=t_{0}$ where  $t_{0}=0$ or $t_{0}=1$ in (\ref{eq:toda-gamma}) we get, respectively,
\begin{equation*}
	-2+(\nu_{n-1}+1)\tilde{\beta}_{n}(t_{0})-(\nu_{n-1}-3)\tilde{\beta}_{n-1}(t_{0})=0\,.
\end{equation*}
Using  $\nu_{n-1}-3=\nu_{n-2}-1, \; n \geq 1,$ and summing from one to $n+1$ we get
\begin{equation*}
	\nu_{n}\tilde{\beta}_{n+1}(t_{0})-\nu_{-1}\tilde{\beta}_{0}(t_{0})+2\sum_{k=1}^{n}\tilde{\beta}_{k}(t_{0})+\tilde{\beta}_{n+1}(t_{0})+\tilde{\beta}_{0}(t_{0})=2n+2\,.
\end{equation*}
Therefore,  $\vartheta_n$ given by (\ref{eq:vartheta}) satisfies (\ref{eq:ic-0})-(\ref{eq:ic-1}).

Taking derivatives in (\ref{eq:toda-gamma}) and using $\nu_{n-1}-3=\nu_{n-2}-1, \; n \geq 1,$ we get
\begin{equation*}
	(2t-1)\partial_{t}\ln \tilde{\gamma}_n+t(t-1)\partial_{t}^2\ln \tilde{\gamma}_n
	=\nu_{n-1}\tilde{\beta}'_n-\nu_{n-2}\tilde{\beta}'_{n-1}+\tilde{\beta}'_{n}+\tilde{\beta}'_{n-1}\,.
\end{equation*}
Summing from one to $n+1$ and using the telescopic sum we get, respectively, 
\begin{align*}
	-\sum_{k=1}^{n+1}\partial_{t}\ln \tilde{\gamma}_k (0) &=
	\nu_n\tilde{\beta}'_{n+1}(0)-\nu_{-1}\tilde{\beta}'_{0}(0)+2\sum_{k=1}^{n}\tilde{\beta}'_k(0)+\tilde{\beta}'_{0}(0)+\tilde{\beta}'_{n+1}(0)\,,
\\
	\sum_{k=1}^{n+1}\partial_{t}\ln \tilde{\gamma}_k (1) &=
	\nu_n\tilde{\beta}'_{n+1}(1)-\nu_{-1}\tilde{\beta}'_{0}(1)+2\sum_{k=1}^{n}\tilde{\beta}'_k(1)+\tilde{\beta}'_{0}(1)+\tilde{\beta}'_{n+1}(1)\,.
\end{align*}
Therefore, $\vartheta_n$ satisfies
\begin{align}
	\vartheta'_n(0)&=-\sum_{k=1}^{n+1}\partial_{t}\ln \tilde{\gamma}_k(0)+(2+\alpha+\beta+\mu)\tilde{\beta}'_0(0)+2{\beta}'_0(0)-\nu_n+\mu\,, \label{eq:vartheta-ic-0}
\\
	\vartheta'_n(1)&=\sum_{k=1}^{n+1}\partial_{t}\ln \tilde{\gamma}_k(1)+(2+\alpha+\beta+\mu)\tilde{\beta}'_0(1)+2{\beta}'_0(1)-\nu_n+\mu\,.\label{eq:vartheta-ic-1}
\end{align}
It remains to compute $\sum_{k=1}^{n+1}\partial_{t}\ln \tilde{\gamma}_k$. Recalling  (\ref{eq:gamma-hn}), we now have
\begin{equation*}
	\ln \tilde{\gamma}_k=\ln \tilde{h}_k-\ln \tilde{h}_{k-1}\,,
\end{equation*}
thus, using the telescopic sum, we get
\begin{equation*}
	\sum_{k=1}^{n+1}\partial_{t}\ln \tilde{\gamma}_k=\partial_{t}\left(\ln \tilde{h}_{n+1}-\ln \tilde{h}_{0}\right)\,.
\end{equation*}
As a consequence,  in terms of the moments in the asymptotic expansion (\ref{eq:expan-S1}) of  $\tilde{f}$, we have
\begin{equation}
	\sum_{k=1}^{n+1}\partial_{t}\ln \tilde{\gamma}_k=\partial_{t}\left(\ln \left(\tilde{\Delta}_{n+2}/\tilde{\Delta}_{n+1}\right)-\ln \tilde{w}_{0}\right)\,,\label{eq:gn}
\end{equation}
where $\tilde{\Delta}_n$ is the corresponding Hankel determinant, given by  $\tilde{\Delta}_n=\det[\tilde{w}_{j+k-2}]_{j, k=1, \dots, n}\,$.
Therefore, in the account of (\ref{eq:vartheta-ic-0})-(\ref{eq:gn}) we obtain (\ref{eq:ic-vartheta0})-(\ref{eq:ic-vartheta1}).
\end{proof}

Furthermore,  $\tilde{w}_0$ and $\tilde{\beta}_0$ may be expressed in terms of the parameters $\beta_0, w_0$, as follows.
\begin{lemma}
The following relations hold:
\begin{gather}
	\tilde{w}_0=-\beta_0^2+\frac{(2+\alpha+\beta+\mu+(2+\alpha+\beta)t)\beta_0-(1+\alpha)t}{3+\alpha+\beta+\mu}\,,
\label{eq:w01-via-beta0}\\
	\frac{(\tilde{\beta}_0-t)(\beta_0-t)\beta'_0}{(\beta_0-t)(\ln w_0)' -  \beta'_0 -\tilde{\beta}'_0}
	=\beta_0^2-\frac{(2+\alpha+\beta+\mu+(2+\alpha+\beta)t)\beta_0-(1+\alpha)t}{3+\alpha+\beta+\mu}\,. \label{eq:tildebeta-beta}
\end{gather}
\end{lemma}
\begin{proof}
Equation (\ref{eq:w01-via-beta0}) follows from the two expressions for the linear term of the polynomial $\tilde{D}$ in (\ref{eq:eg-pol-D1}), say $d_{1,1}$:
on the one hand, the expression that is obtained by using the asymptotic expansion (\ref{eq:expan-S1}) in (\ref{eq:ric-f1-x}) gives
\begin{equation*}
	d_{1,1}=-(3+\alpha+\beta+\mu)\tilde{w}_0\,,
\end{equation*}
whilst on the other hand, from (\ref{eq:eg-pol-D1}) we have
\begin{equation*}
	d_{1,1}=(3+\alpha+\beta+\mu)\beta_0^2-(2+\alpha+\mu+(2+\alpha+\beta)t)\beta_0+(1+\alpha)t\,.
\end{equation*}

The coefficient of $x^n$ in the equations in the $(1,2)$ entry of  (\ref{eq:eg-sylvester-Yn-t}) gives us, by virtue of (\ref{eq:another-hat-thetan}),
\begin{equation*}
	\tilde{w}_0=\frac{[({\beta}_0-t)(\ln w_0)'+\mu ]({\beta}_0-t)(\tilde{\beta}_0-t)}{2(\beta_0-t)(\ln w_0)' + \mu -\tilde{\beta}'_0}\,.
\label{eq:w0-aux}
\end{equation*}
Using (\ref{eq:der-w0}) in the previous equation we get 
\begin{equation}
	\tilde{w}_0=-\frac{(\tilde{\beta}_0-t)(\beta_0-t)\beta'_0}{(\beta_0-t)(\ln w_0)' - \beta'_0 -\tilde{\beta}'_0}\,.
\label{eq:w-0-1}
\end{equation}
The equation (\ref{eq:tildebeta-beta}) is then a consequence of (\ref{eq:w-0-1}) and (\ref{eq:w01-via-beta0}).
\end{proof}

\subsubsection{Relation with the Semi-classical Case}

Now that the Painlev\'e equation P$_{\textrm{VI}}$ for the sequences of orthogonal polynomials corresponding to $\tilde{f}$  was derived, and given that
$\tilde{f}$ is a M\"{o}bius transformation of $f$ (see (\ref{eq:til-f})),
\begin{equation*}
	\tilde{f}(x)=x-\beta_0-\frac{w_0}{f(x)}\,,
\end{equation*}
a natural question that can be posed is to relate the P$_{\textrm{VI}}$ solution just derived with the Painlev\'e solution that governs the sequences of polynomials corresponding to $f$.
Recall that $f$ is the Stieltjes transform of the {deformed} Jacobi weight $w$ defined by (\ref{eq:jacobi-w}).
The sequences of orthogonal polynomials corresponding to $w$,
 indeed taken orthonormal, $\{p_n\}_{n \geq 0}$,  were subject of analysis, e.g.,  in \cite{magnus-jcam}.
In \cite[pp. 227-228]{magnus-jcam} it is deduced that $q$ defined by the only zero of the $(1,2)$ entry of the matrix of the differential system
\begin{equation*}
	\frac{d}{dx}\left[
	\begin{array}{cc}
		p_{n} & \varepsilon_{n+1}/w \\
		p_{n-1} & \varepsilon_n/w
	\end{array}
	\right]=
	\frac{1}{A}
	\begin{bmatrix}
		\Omega_{n}-C/2 & -a_n\Theta_{n} \\
		a_n\Theta_{n-1} &\; -\Omega_{n}-C/2
	\end{bmatrix}\left[
	\begin{array}{cc}
		P_{n+1} & \varepsilon_{n+1}/w \\
		P_n & \varepsilon_n/w
	\end{array}\right]\,
\end{equation*}
(here, we are using the notation $A, C$ for the polynomials $W, V$ of \cite{magnus-jcam}, respectively, and $a_n$ are the recurrence relation coefficients defined by $a_n^2=\gamma_n$),
satisfies a P$_{\textrm{VI}}$ equation (\ref{eq:P6}) with parameters
$(\delta_1, \delta_2, \delta_3, \delta_4)= (\frac{1}{2}, -\frac{1}{2}\alpha^2, \frac{1}{2}\beta^2, \frac{1}{2}(1-\mu^2)).$
Hence, the Laguerre-Hahn si\-tua\-tion {\it{versus}} the semi-classical one differs in the boundary conditions.
On this topic we refer the interested reader to \cite[Sec. 2]{F-W1}, and the techniques used in the proof of Proposition 2.1, pages 12219-12220,
relying in certain identities resulting from expansions of Toeplitz determinants.

\section{Conclusions}\label{sec:4}
\setcounter{equation}{0}
We have demonstrated with an important example that the natural deformation of a Laguerre-Hahn class system, which is directly related to the deformed Jacobi weight,
is governed by solutions to the sixth Painlev\'e equation with a parameter (or monodromy exponents) set identical to the solution for this semi-classical weight.
The two solutions only differ in their boundary data, or equivalently the off-diagonal elements of their connection matrices.

However there is an interesting question posed by our result. We have only an implicit characterisation of this system by the fact that its Stieltjes transform satisfies a
certain Riccati differential equation - we have not given an explicit solution for this transform nor for the weight itself.
The construction of the transform or better still that of the weight, even in special cases, would be a major achievement.
And we would speculate that this problem bears some relation to the problem of constructing the weight for the associated Jacobi polynomials,
a task performed by Grosjean \cite{gros2} and Wimp \cite{Wim_1987}, via a deformation of this situation.
From a different point of view Nevai \cite{Nev_1984} has given a construction for associated polynomials in a general setting but the relationship of his method to the Riccati equation has yet to be explored.

Interestingly there is an apparent connection of this theme with the diagonal spin-spin correlations of the Ising model on the anisotropic square lattice.
These are known to be characterised by a particular solution of P$_{\rm VI}$, through the work of Jimbo and Miwa \cite{JM_1980}, and are also related to the Picard solution.
The correlations can be written as Toeplitz determinants of the first, second and third kinds of elliptic integrals with the deformation variable $t$ corresponding to the temperature,
and these determinants possess a simple algebraic symbol and are therefore semi-classical - see \cite{McCoy+Wu_1973}.
An observation made by Mangazeev and Guttmann \cite{MG_2010} was that these solutions have a one parameter generalisation
(denoted by $\lambda$) which enters only through the boundary conditions.
Assuming the Toda lattice equations appropriate to P$_{\rm VI}$ they observed that for explicit examples of small spin separations they could evaluate these as
Jacobi Elliptic functions for general $t$ and $\lambda$, with the value $ \lambda=1 $ recovering the standard model.
One of the authors (NSW) has confirmed that these explicit examples satisfy the other standard differential-difference relations of P$_{\rm VI}$.
Furthermore, it has been shown that the deformation differential equation for the Stieltjes function has the Riccati form with the quadratic term only if $ \lambda \neq 1 $.
It is also known that these solutions have a perturbation expansion about $ \lambda=0 $ which coincides with the form-factor expansion of the correlations.
However it is not known what symbol or weight yields these $\lambda$-correlations, nor the general formula for these at arbitrary separation.

\section*{Acknowledgments}
The work of MNR is partially supported by CMA-UBI (grants UID-B-MAT/00212/2020 and UID-P-MAT/00212/2020), funded by FCT.
The work of MNR is also partially supported by Dirección General de Investigación e Innovación, Consejería de Educación e Investigación of the Comunidad de Madrid (Spain), and Universidad de Alcalá under grant CM/JIN/2021-014, Proyectos de I+D para Jóvenes Investigadores de la Universidad de Alcalá 2021.
NSW would like to acknowledge the support of grant UID/MAT/00324/2019  provided by  FCT/Portugal
during the visit to CMUC, Universidade de Coimbra, in October-November 2019.
The authors wish to thank the referees for their careful revision of the manuscript
and in particular their comments and suggestions which have improved the presentation of the manuscript.

\end{document}